\RequirePackage{amsmath}
\documentclass[12pt]{llncs}

\usepackage[T1]{fontenc}
\usepackage{amsfonts}
\usepackage{amssymb}
\usepackage{tikz}
\usetikzlibrary{automata}
\usetikzlibrary{shadows}
\usetikzlibrary{arrows}
\usetikzlibrary{shapes}
\usetikzlibrary{decorations.pathmorphing}
\usetikzlibrary{fit}
\usepackage{algorithm}
\usepackage{algorithmic}
\usepackage[margin=2cm]{geometry}
\usepackage{dsfont}

\usepackage{enumitem}
\def\slash{\relax\ifmmode\delimiter"502F30E\mathopen{}\else\@old@slash\fi}

\tikzstyle{every picture}=[
  >=stealth', shorten >=1pt, node distance=1.44cm,auto,bend angle=45,initial text=,
  every state/.style={inner sep=0.75mm, minimum size=1mm},font=\scriptsize,
]

\def\mmod#1#2{\left[#1\right]_{#2}}

\def\sc{\mathrm{sc}}
\def\oo{\star}

\title{State complexity of catenation combined with boolean operations}
\author{
    Pascal Caron
    \and Jean-Gabriel Luque
    \and Bruno Patrou
    \thanks{\{Pascal.Caron, Jean-Gabriel.Luque, Ludovic.Mignot, Bruno.Patrou\}@univ-rouen.fr}
}
\institute{LITIS, Universit\'e de Rouen,\\ Avenue de l'Universit\'e,\\ 76801 Saint-\'Etienne du Rouvray Cedex,\\ France}
\def\sat{\mathrm{sat}}
\def\split{\mathrm{split}}
\begin{document}

\maketitle
\begin{abstract}
We exhaustively investigate  possible combinations of a boolean operation together with a catenation. In many cases we prove and improve some conjectures by Brzozowski. 
For each family of operation, we endeavour to provide a common witness with a small size alphabet. 
\end{abstract}

\section{Introduction}

State complexity is a very active research area. It aims to determine the maximal size of a minimal automaton recognizing a language belonging to a given class. State complexity can be studied from the deterministic as well as non-deterministic point of view. Here, we only consider the deterministic case. Then, the state complexity of a regular language is the states number of its minimal DFA (Deterministic Finite Automaton).  The state complexity of a regular operation allows to compute the maximal size of any DFA obtained by applying this operation over regular languages, knowing their respective state complexities. Such operations can be elementary (see, as one of the first reference in this domain, \cite{Mas70}) or the result of some combinations (see, for example, \cite{GSY08}, \cite{CGKY11} or \cite{JO11}). Sometimes, the computation of state complexities needs to use  combinatorial tools, as in \cite{CLMP15}. To have an expanded view of the domain, it is useful to refer to the survey  \cite{GMRY15}.

 In \cite{Brz13}, J. Brzozowski shows that a particular family of DFAs, that we call Brzozowski automata, is used to produce witnesses in a very large number of cases.
This family of DFA are such that the letters must play one of the four following roles:
	 a total cycle,
	 a transposition,
	 a contraction or
	 the identity.
The first three  roles are used to maximize the semigroup of transformations.
 
We illustrate the power of this approach by revisiting and completing the picture concerning the combination of catenation with any boolean operation. We recall the known results and some Brzozowski conjectures. Then we prove the conjectures and some new results.

We give a complete panorama of state complexity of each possible  combination involving the catenation and/or a boolean operation.
As noticed in \cite{CLMP15}, it is sufficient to focus on the three operators $\cap$, $\cup$ and $\oplus$ to produce the desired results. The possible combinations are
\begin{enumerate}
	\item\label{cas5} $ABC$
	\item\label{cas4} $A\circ B\circ C$
	\item\label{cas3} $A(B\circ C)$
	\item\label{cas2} $(A\circ B)C$
	\item\label{cas1} $(AB)\circ C$ (which immediately implies $A\circ(BC)$ since $\circ$ is commutative)
\end{enumerate}
For each of these combinations, some results are already known.
\begin{enumerate}
\item The bound is given in \cite{GY09} and a $3$-letters witness is given in \cite{CLP16}.  
\item When $\circ$ is restricted to $\cap$ and $\cup$, this is a particular case from a more general study done in \cite{EGLY09} from which  a $6$-letters witness is deduced. 
\item  The bound is given and reached in  \cite{CGKY11} with a $3$-letters witness when   $\circ$  is $\cup$ and a $4$-letters witness when $\circ$ is $\cap$. When $\circ$ is $\oplus$ the bound is given in \cite{CLMP15} and a $4$-letters witness is provided.
\item The bound is given and reached in  \cite{CGKY12}   with  $4$-letters witnesses when   $\circ$  is $\cup$  or $\cap$. 
\item  The bound is given and reached in  \cite{CGKY12} with  $3$-letters witnesses when   $\circ$  is $\cup$  or $\cap$.
\end{enumerate}
We improve some of the previous results as follows.
\begin{enumerate}
\setcounter{enumi}{1}
\item For the $9$ cases where each $\circ$ is replaced by $\cap$, $\cup$ or $\oplus$, we produce a common $2$-letters witness.
\item Conjectures $18$ and $19$ of \cite{Brz13} provide a common $4$-letters Brzozowski witness for  $\cup$ and $\cap$. We improve this result by giving a common $3$-letters witness for $\cup$ and $\cap$. We also show that this witness does not suit to the case of $\oplus$. In this last case, we conjecture that $4$-letters are needed.
\item We provide a common $3$-letters Brzozowski witness for the $3$ operations ($\cup$, $\cap$ and $\oplus$), resolving  Conjecture $6$ of \cite{Brz13}.
\item We provide a common $2$-letters Brzozowski witness for the $3$ operations ($\cup$, $\cap$ and $\oplus$), improving Conjecture $5$ of \cite{Brz13}.
\end{enumerate}

\section{Preliminaries}\label{preliminaries}

For any integer $i\in \mathbb Z$, any $p\in \mathbb N\setminus \{0\}$, we set  $\mmod{i}{p}=\mathrm{min}\{j\mid j\geq 0 \wedge j\equiv i (p)\}$. 
Let $\Sigma$ 
denotes a finite alphabet. A word $w$ over  $\Sigma$ is a finite sequence of symbols of $\Sigma$. The length of  $w$, denoted by $|w|$ is the number of occurrences of symbols of $\Sigma$ in $w$. For  $a\in \Sigma$, we denote by $|w|_a$ the number of $a$ in $w$.  The set of all finite words over $\Sigma$ is denoted by $\Sigma ^*$.  The empty word is denoted by  $\varepsilon$. A language is a subset of $\Sigma^*$. The set of subsets of a finite set $A$ is denoted by $2^A$ and $|A|$ denotes the cardinality of $A$.  We denote by $\uplus$ the union of disjoint sets. The symbol $\circ$ denotes any binary boolean operation on languages. In the following, by abuse of notation, we 
often write $q$ for any singleton $\{q\}$.

A  finite automaton (FA) is a $5$-tuple $A=(\Sigma,Q,I,F,\cdot)$ where $\Sigma$ is the input alphabet, $Q$ is a finite set of states, $I\subset Q$ is the set of initial states, $F\subset Q$ is the set of final states and $\cdot$ is the transition function from  $Q\times \Sigma$ to $2^Q$. An FA is deterministic  and complete (DFA) if $|I|=1$ and for all $q\in Q$, for all $a\in \Sigma$, $|q\cdot a=1$. 
 The transition function is  extended to any word  by $q\cdot a w=\bigcup _{q'\in q\cdot a} q'\cdot w$ and  $q\cdot \varepsilon=q$ for any symbol $a$   of $\Sigma$ and  any word $w$  of $\Sigma^*$. 
  For convenience,  we sometimes use the notation $q\xrightarrow{w} q'$ to denote that $q'\in q\cdot w$.

The dual operation is defined by 
 $w\cdot q=\{q'\mid q\in q'\cdot w\}$.  We  extend the dot notation to any set of states $S$ by $S\cdot w=\bigcup_{s\in S}s\cdot w$ and $w\cdot S =\bigcup_{s\in S} w\cdot s$.  A word $w\in \Sigma ^*$ is recognized by an FA $A$ if $I\cdot w\cap F\neq \emptyset$. 

The language recognized by an FA $A$ is the set $L(A)$ of words recognized by $A$. 
Two automata are said to be equivalent if they recognize the same language.  

Let $D=(\Sigma,Q_D,i_D,F_D,\cdot)$ be a DFA.
Two states $q_1,q_2$ of  $D$ are equivalent if for any word $w$ of $\Sigma^*$, $q_1\cdot w\in F_D$ if and only if $q_2\cdot w\in F_D$. Such an equivalence is denoted by $q_1\sim q_2$. A DFA is  minimal if there does not exist any equivalent  complete DFA  with less states and it is well known that for any DFA, there exists a unique minimal equivalent one \cite{HU79}. Such a minimal DFA  can be  obtained from $D$ by computing the accessible part of the automaton $D\slash \sim=(\Sigma,Q_D\slash \sim,[i_D],F_D\slash \sim,\cdot)$ where for any $q\in Q_D$, $[q]$ is the $\sim$-class of the state $q$ and satisfies the property  $[q]\cdot a=[q\cdot a]$, for any $a\in \Sigma$. In a minimal DFA, any two distinct states are pairwise inequivalent.

 The state complexity of a regular language $L$ denoted by $\sc(L)$ is the number of states of its minimal DFA. 
  Let ${\cal L}_n$ be the set of languages of state complexity $n$. The state complexity of a unary operation $\otimes$ is the function $\sc_{\otimes}$ associating with an integer $n$ the maximum of the state complexities of $(\otimes L)$ for $L\in {\cal L}_n$.
  A language $L\in {\cal L}_n$ is a witness (for $\otimes$) if  $\sc(\otimes L)=\sc_{\otimes}(n)$. 
  This can be generalized, and the state complexity of a $k$-ary operation $\otimes$ is the $k$-ary function which associates with any tuple $(n_1,\ldots,n_k)$ the integer $\mathrm{max}\{\sc(\otimes(L_1,\ldots,L_k))|L_i\in\mathcal{L}_{n_i},\forall i\in[1,k]\}$. Then, a witness is a tuple $(L_{1},\ldots,L_{k})\in({\cal L}_{n_1}\times \cdots  \times{\cal L}_{n_k})$ such that $\sc(\otimes(L_{1},\ldots,L_{k}))=\sc_{\otimes}(n_1,\ldots,n_k)$. 
  An important research area consists in finding witnesses for any $(n_1,\ldots ,n_k)\in \mathbb{N}^k$.
In the aim to manipulate combinations of binary operators, we introduce the following notation. A binary operator $\otimes$ is also denoted by $\oo\otimes\oo$ and we extend the notation for any combination of binary operators.
   For example, the ternary operation  defined for any three languages $L_1, L_2, L_3$ by  $L_1\cdot(L_2\cup L_3)$ is denoted by $\oo \cdot (\oo \cup \oo)$. Let $h$ be its state complexity. 
  Let $f,g$ be the  respective state complexity of $\oo\cdot\oo$ and $\oo\cup\oo$. For any three integers $n_1,n_2,n_3$, it holds $h(n_1,n_2,n_3)\leq f(n_1,g(n_2,n_3))$ \cite{GY09}. In fact, applying the union on a witness does not produce a good candidate for a witness for catenation. Indeed, about half of the states of the obtained DFA are final which contradicts the fact that a good candidate must have only one final state \cite{JJS05}.
  
\subsection{Brzozowski witnesses\label{Brzozowski witnesses}}

In \cite{Brz13}, Brzozowski defines  a family of languages  that turns to be universal witnesses for several operations. The automata denoting these languages are called \textit{Brzozowski automata}.
We  need some background to define these automata. We 
follow the terminology of \cite{GM08}. Let $Q=\{0,\ldots, n-1\}$ be a set. A \textit{transformation} of the set $Q$ is a mapping of $Q$ into itself. If $t$ is a transformation and $i$ an element of $Q$, we denote by $it$ the image of $i$ under $t$. A transformation of $Q$ can be represented by $t=[i_0, i_1, \ldots i_{n-1}]$ which means that $i_k=kt$ for each $0\leq k\leq n-1$ and $i_k\in Q$. A \textit{permutation} is a bijective transformation on $Q$. The \textit{identity} permutation of $Q$ is denoted by $\mathds{1}$. A \textit{cycle} of length $\ell\leq n$  is a permutation $c$, denoted   by $(i_0,i_1,\ldots i_{\ell-1})$, on a subset $I=\{i_0,\ldots ,i_{\ell-1}\}$ of $Q$  where  $i_kc=i_{k+1}$ for $0\leq k<\ell-1$ and $i_{\ell-1}c=i_0$.  A \textit{transposition} $t=(i,j)$ is a permutation on $Q$ where $it=j$ and $jt=i$ and for every  elements $k\in Q\setminus \{i,j\}$, $kt=k$.  A \textit{contraction}  $t=\left(\begin{array}{r}i\\j\end{array}\right)$ is a transformation where  $it=j$ and  for every  elements $k\in Q\setminus \{i\}$, $kt=k$.
Then,  a Brzozowski automaton is a complete  DFA $(\Sigma, Q=\{0,\ldots , n-1\}, 0, F=\{n-1\}, \cdot)$, where any letter of $\Sigma$ induces one of the transformation among transposition, cycle over $Q$, contraction and identity.

To define shortly such a DFA, we introduce the following definition:
\begin{definition}
A Brzozowski automaton  ${\mathcal X}_n(\sigma_1,\sigma_2,\sigma_3;\Sigma\setminus \{\sigma_1,\sigma_2, \sigma_3\})=(\Sigma, \{0,\ldots, n-1\},0,\{n-1\},\cdot)$ is a DFA in which $\sigma_1, \sigma_2, \sigma_3 \in \Sigma\cup \{-\}$ and each  symbol induces a transformation:
\begin{itemize}
\item the letter $\sigma_1\neq -$ induces the $n$-cycle $(0,\ldots, n-1)$,
\item the letter $\sigma_2\neq -$ induces the transposition $(0,1)$,
\item the letter  $\sigma_3\neq -$ induces the contraction $\left(\begin{array}{c}1\\0\end{array}\right)$,
\item every other letter of $\Sigma$ induces the identity on $Q$.
\end{itemize}
\end{definition}

  Let $\Sigma= \{a,b,c,d\}$. As an example of Brzozowski automata (see Figure \ref{Brzo}), let \label{Brzo-def} ${\cal X}_n(a,-,c;\{b,d\})=(\Sigma,Q_n,0,\{n-1\},\cdot)$ where $Q_n=\{0,1,\ldots ,n-1\}$, the symbol $a$ acts as the cycle  $(0,1,\ldots, n-1)$, $c$ acts as the contraction $\left(\begin{array}{r}1\\0\end{array}\right)$ and   $b,d$ act as $\mathds{1}$. 

 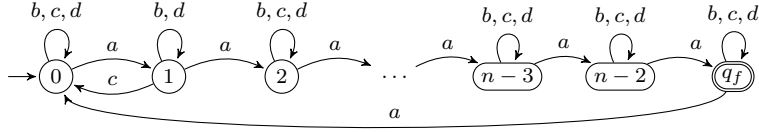
\begin{figure}[htb]
	\centerline{
		\begin{tikzpicture}[node distance=1.5cm, bend angle=25]
			\node[state,initial] (0)  {$0$};
			\node[state] (1) [right of=0] {$1$};
			\node[state] (2) [right of=1] {$2$};
			\node (etc1) [right of=2] {$\ldots$};
			\node[state, rounded rectangle] (m-3) [right of=etc1] {$n-3$};
			\node[state, rounded rectangle] (m-2) [right of=m-3] {$n-2$};
			\node[state,accepting, rounded rectangle] (m-1) [right of=m-2] {$q_f$};
			\path[->]
        (0) edge[bend left] node {$a$} (1)
        (1) edge[bend left] node {$a$} (2)
        (2) edge[bend left] node {$a$} (etc1)
        (etc1) edge[bend left] node {$a$} (m-3)
        (m-3) edge[bend left] node {$a$} (m-2)
        (m-2) edge[bend left] node {$a$} (m-1)
        (m-1) edge[out=-115, in=-65, looseness=.2] node[above] {$a$} (0)
		    (0) edge[out=115,in=65,loop] node {$b,c, d$} (0)
		    (1) edge[out=115,in=65,loop] node {$b,d$} (1)
		    (2) edge[out=115,in=65,loop] node {$b, c, d$} (2)
		    (m-3) edge[out=115,in=65,loop] node {$b, c, d$} (m-3)
		    (m-2) edge[out=115,in=65,loop] node {$b, c, d$} (m-2)
		    (m-1) edge[out=115,in=65,loop] node {$b, c, d$} (m-1)
        (1) edge[bend left] node[above] {$c$} (0)
			;
\end{tikzpicture}
}
\caption{The automaton ${\cal X}_n(a,-,c;\{b,d\})$}\label{Brzo}
\end{figure}	

For convenience, in the following of the paper, we  identify ${\mathcal X}_n$ and $L({\mathcal X}_n)$.

\subsection{Construction algorithms}
We define an operation on automata allowing us to compute a DFA for the catenation of two DFAs.

\begin{definition}\label{CatenationDFAs}
 Let $A=(\Sigma,Q_A,i_A,F_A,\cdot_A)$  and $B=(\Sigma,Q_B,i_B,F_B,\cdot_B)$ be two DFAs. We define the DFA $A\cdot B=(\Sigma, Q, i,F, \cdot)$ as follows : 
\begin{itemize}
\item $Q=\{(p,S)\mid p\in Q_A, S\subset Q_B\}$
\item $i=\left\{\begin{array}{ll}(i_A,\emptyset)&\mbox{ if }i_A\not\in F_A\\(i_A,\{i_B\})&\mbox{ otherwise}\end{array}\right.$
\item $F=\{(p,S)\mid S\cap F_B\neq \emptyset \}$
\item $(p,S)\cdot a=\left\{\begin{array}{ll}(p\cdot a, S\cdot a)&\mbox{ if }p\cdot a\not\in F_A\\(p\cdot a, S\cdot a\cup \{i_B\})&\mbox{ otherwise}\end{array}\right.$
\end{itemize}
\end{definition}
We now define an operation on automata allowing us to compute a DFA for any  boolean operation over two DFAs.

\begin{definition}\label{CartesianDFAs}
 Let $A=(\Sigma,Q_A,i_A,F_A,\cdot_A)$  and $B=(\Sigma,Q_B,i_B,F_B,\cdot_B)$ be two DFAs. We define the DFA $A\circ B=(\Sigma, Q, i,F, \cdot)$ as follows : 
\begin{itemize}
\item $Q=\{(p,q)\mid p\in Q_A, q\in Q_B\}$
\item $i=(i_A,i_B)$
\item $F=\{(p,q)\mid p\in F_A \}\circ \{(p,q)\mid q\in F_B \}$
\item $(p,q)\cdot a= (p\cdot a, q\cdot a)$
\end{itemize}
\end{definition}
It is easy to verify the following lemma :
\begin{lemma}
$L(A\cdot B)=L(A)\cdot L(B)$ and $L(A\circ B)=L(A)\circ L(B)$.
\end{lemma}
These constructions can be combined in several ways.   Table \ref{combop} summarizes the different forms of states one can have.

\begin{table}[h]
$$\begin{array}{|c|c|}
\hline
(A\cdot  B)\cdot C&(i,S_1,S_2)\\
\hline
(A\cdot  B)\circ C&(i,S_1,k)\\
\hline
(A\circ  B)\cdot C&(i,j,S_2)\\
\hline
A\cdot  (B\circ C)&(i,T)\\
\hline
(A\circ_1  B)\circ_2 C&(i,j,k)\\
\hline
\end{array}$$
\caption{Forms of states for  combined operations where $i\in Q_A, j\in Q_B, k\in Q_C, S_1\subset Q_B, S_2\subset Q_C, T\subset Q_B\times Q_C$}\label{combop} 
\end{table} 
In the following of the paper, the name of the state is  considered modulo the size of the automaton it belongs to. For instance, the state $(i,j,S_2)$ is in fact the state $([i]_m,[j]_n,S_2)$ where $m$ and $n$ are the respective  number of states of $A$ and $B$. 

\subsection{Tableaux}
About the construction  $A\cdot(B\circ C)$ whose states are of the form $(i,T)$, 
 the set $T$ can be seen as a tableau with $n$ rows and $p$ columns where any cell $(j,k)$ is marked if and only if the couple of states $(j,k)$ is in $T$ (see Figure \ref{tableau}). In the following, for simplicity, when the dimensions are fixed, we 
assimilate a tableau with the set of its marked cells.

  \begin{figure}[H]
    \centerline{ 
      \begin{tikzpicture}[scale=0.5]   
	    \foreach \x in {1,...,7} {
	      \foreach \y in {1,...,6} {
	        \pgfmathparse{\x+1} \let\z\pgfmathresult
	        \pgfmathparse{\y+1} \let\t\pgfmathresult
	        \draw[fill=gray!40] (\x,\y) rectangle (\x+1,\y+1);
	      }
	    }  
	    \foreach \x/\y in {3/3,6/3,6/5} {
	        \pgfmathparse{\x+1} \let\z\pgfmathresult
	        \pgfmathparse{\y+1} \let\t\pgfmathresult
	        \draw[fill=white] (\x,\y) rectangle (\x+1,\y+1);
	        \draw (\x,\y) -- (\z,\t);
	        \draw (\z,\y) -- (\x,\t);
	    }
      \end{tikzpicture}
    }
    \caption{The tableau corresponding to $T=\{(3,2),(1,5),(3,5)\}$ with $n=6$ and $p=7$.}
    \label{tableau}    
  \end{figure}
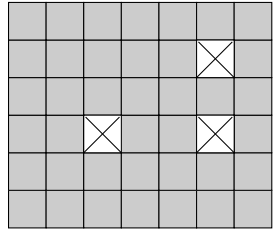 

Since the state complexity of catenation is $ \sc_\bullet(m,m')=(m-1)2^{m'}+2^{m'-1}$ and the state complexity of a binary boolean operation
$\circ$
is bounded by $\sc_\circ(n,p)=np$ (see \cite{YZS94}), their composition allows to bound
the state complexity of  $A\cdot(B\circ C)$ by $(m-1)2^{np}+2^{np-1}$. This bound is reached when $\circ=\cap$ \cite{CGKY11}.

 The state complexity for the combination of catenation with union ($A\cdot(B\cup C)$) has been studied in \cite{CGKY11} but it can be reinterpreted using the tableaux defined previously. Let $(i,T)$ and $(i,T')$ be two distinct states such that the couples $(x,x')$ and $(y,y')$ are in $T'$ and $T=T'\cup \{(x,y')\}$. Then the two states $(i,T)$ and $(i,T')$ are  equivalent. Indeed,  to separate  these states, one has to find a word $w$ such that $(1)$ $T'\cdot w$ is equal to a set of couples which members are both non-final and $(2)$ $(x,y')\cdot w$ leads to a couple of states at least one of the two is final. The fact that $x\cdot w$ or $y'\cdot w$ is final  contradicts $(1)$. So $(i,T)$ and $(i,T')$ are equivalent.

 Such equivalent states have tableaux with specific patterns. Indeed, the tableaux for $T$ and $T'$ contain the pattern of Figure \ref{fig two tab}(a) and Figure \ref{fig two tab}(b) respectively. None of them can be distinguished from the pattern of  Figure \ref{fig two tab}(c). So the number of equivalent states is the number of indistinguishable tableaux represented by the patterns of  Figure \ref{fig two tab}. The number of tableaux not containing patterns of Figure \ref{fig two tab}(a) or Figure \ref{fig two tab}(b) is $(2^n-1)(2^p-1)+1$.
 
 \begin{figure}[htb]
 
	\begin{minipage}{0.32\linewidth}
	\centerline{
  \begin{tikzpicture}[scale=0.4]   
	    \foreach \x in {1,...,6} {
	      \foreach \y in {1,...,5} {
	        \pgfmathparse{\x+1} \let\z\pgfmathresult
	        \pgfmathparse{\y+1} \let\t\pgfmathresult
	        \draw[fill=gray!40] (\x,\y) rectangle (\x+1,\y+1);
	      }
	    }  
	    \foreach \x/\y in {2/2,5/4} {
	        \pgfmathparse{\x+1} \let\z\pgfmathresult
	        \pgfmathparse{\y+1} \let\t\pgfmathresult
	        \draw[fill=white] (\x,\y) rectangle (\x+1,\y+1);
	        \draw (\x,\y) -- (\z,\t);
	        \draw (\z,\y) -- (\x,\t);
	    }	   
  	    \foreach \x/\y in {5/2,2/4} {
	        \pgfmathparse{\x+1} \let\z\pgfmathresult
	        \pgfmathparse{\y+1} \let\t\pgfmathresult
	        \draw[fill=white] (\x,\y) rectangle (\x+1,\y+1);
	    }	   
    \end{tikzpicture}
}
\end{minipage}
	\begin{minipage}{0.32\linewidth}
	\centerline{
  \begin{tikzpicture}[scale=0.4]   
	    \foreach \x in {1,...,6} {
	      \foreach \y in {1,...,5} {
	        \pgfmathparse{\x+1} \let\z\pgfmathresult
	        \pgfmathparse{\y+1} \let\t\pgfmathresult
	        \draw[fill=gray!40] (\x,\y) rectangle (\x+1,\y+1);
	      }
	    }  
	    \foreach \x/\y in {2/2,5/2,5/4} {
	        \pgfmathparse{\x+1} \let\z\pgfmathresult
	        \pgfmathparse{\y+1} \let\t\pgfmathresult
	        \draw[fill=white] (\x,\y) rectangle (\x+1,\y+1);
	        \draw (\x,\y) -- (\z,\t);
	        \draw (\z,\y) -- (\x,\t);
 }
  	    \foreach \x/\y in {2/4} {
	        \pgfmathparse{\x+1} \let\z\pgfmathresult
	        \pgfmathparse{\y+1} \let\t\pgfmathresult
	        \draw[fill=white] (\x,\y) rectangle (\x+1,\y+1);
}	    
      \end{tikzpicture}
}
\end{minipage}
	\begin{minipage}{0.32\linewidth}
	\centerline{
  \begin{tikzpicture}[scale=0.4]   
	    \foreach \x in {1,...,6} {
	      \foreach \y in {1,...,5} {
	        \pgfmathparse{\x+1} \let\z\pgfmathresult
	        \pgfmathparse{\y+1} \let\t\pgfmathresult
	        \draw[fill=gray!40] (\x,\y) rectangle (\x+1,\y+1);
	      }
	    }  
	    \foreach \x/\y in {2/2,5/2,2/4, 5/4} {
	        \pgfmathparse{\x+1} \let\z\pgfmathresult
	        \pgfmathparse{\y+1} \let\t\pgfmathresult
	        \draw[fill=white] (\x,\y) rectangle (\x+1,\y+1);
	        \draw (\x,\y) -- (\z,\t);
	        \draw (\z,\y) -- (\x,\t);
	    }
      \end{tikzpicture}
}
\end{minipage}
\caption{Three indistinguishable tableaux (a), (b), (c), for the union operator.}\label{fig two tab}
\end{figure}
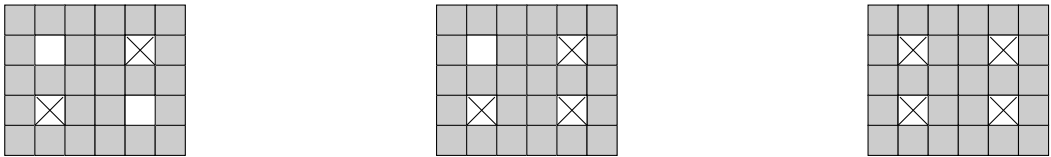	
	
 Indeed, one has to choose among $n$ rows and $p$ columns (at least one of each) and mark every cell at the intersection of the chosen rows and columns ($(2^n-1)(2^p-1)$) plus one configuration with no cell marked. We also have to count the same tableaux but with the cell $(0,0)$ marked ($2^{n-1}2^{p-1}$ tableaux). Combined with the state complexity of catenation, these observations 
lead
to the state complexity $(m-1)((2^n-1)(2^p-1)+1)+2^{n-1}2^{p-1}=(m-1)(2^{n+p}-2^n-2^p+2)+2^{n+p-2}$ of $A\cdot(B\cup C)$.
 
  It is easy to check that there exist DFAs $A$, $B$ and $C$ such that there are no indistinguishable tableaux for $A\cdot(B\cap C)$, the state complexity of catenation combined with intersection coincides with the bound.

As for the union, some particular states are necessarily equivalent for $A\cdot(B\oplus C)$. Let $(i,T)$ and $(i,T')$ be two distinct states such that the couples $(x,x')$, $(x,y')$ and $(y,y')$ are in $T'$ and $T=T'\cup \{(y,x')\}$. Then the two states $(i,T)$ and $(i,T')$ are equivalent.
  Indeed, if a word $w$ separates  $(i,T)$ and $(i,T')$, then $w$ sends $y$ in $F_B$ or $x'$ in $F_C$ but not both,  sending $(i,T)$ to a final state of $A\cdot(B\oplus C)$. This cannot be achieved without sending  $(i,T')$ to a final state of $A\cdot(B\oplus C)$, thus contradicting the separation by $w$. 

  Such equivalent states imply indistinguishable tableaux as described below.
  Four distinct marked cells $s_1$, $s_2$, $s_3$ and $s_4$ define a \emph{rectangle} if there exist four integers $x$, $x'$, $y$ and $y'$ such that $\{s_1,s_2,s_3,s_4\}=\{x,y\}\times\{x',y'\}$. 
  Three distinct marked cells $s_1$, $s_2$ and $s_3$ form a \emph{right triangle} if there exists an unmarked cell $s_4$ such that $s_1$, $s_2$, $s_3$ and $s_4$ form a rectangle (See Figure~\ref{fig rect} and Figure~\ref{fig tri}). 
  
  A tableau $T$ is \emph{saturated} if it is the union of all its equivalent tableaux. Informally, to saturate a tableau, it is sufficient to complete the tableau by marking  the missing cells for   the considered operation.
  
  \begin{minipage}{0.45\linewidth}
    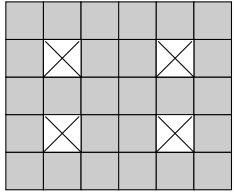
\begin{figure}[H]
      \centerline{
        \begin{tikzpicture}[scale=0.5]   
	      \foreach \x in {1,...,6} {
	        \foreach \y in {1,...,5} {
	          \pgfmathparse{\x+1} \let\z\pgfmathresult
	          \pgfmathparse{\y+1} \let\t\pgfmathresult
	          \draw[fill=gray!40] (\x,\y) rectangle (\x+1,\y+1);
	        }
	      }  
	      \foreach \x/\y in {2/2,5/4,2/4,5/2} {
	        \pgfmathparse{\x+1} \let\z\pgfmathresult
	        \pgfmathparse{\y+1} \let\t\pgfmathresult
	        \draw[fill=white] (\x,\y) rectangle (\x+1,\y+1);
	        \draw (\x,\y) -- (\z,\t);
	        \draw (\z,\y) -- (\x,\t);
	      }
        \end{tikzpicture}
      }
      \caption{A rectangle.}
      \label{fig rect}
    \end{figure}
  \end{minipage}
  \hfill  
  \begin{minipage}{0.45\linewidth}
    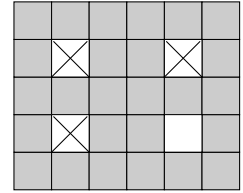
\begin{figure}[H]
      \centerline{
        \begin{tikzpicture}[scale=0.5]   
	      \foreach \x in {1,...,6} {
	        \foreach \y in {1,...,5} {
	          \pgfmathparse{\x+1} \let\z\pgfmathresult
	          \pgfmathparse{\y+1} \let\t\pgfmathresult
	          \draw[fill=gray!40] (\x,\y) rectangle (\x+1,\y+1);
	        }
	      }  
	      \foreach \x/\y in {2/2,5/4,2/4} {
	        \pgfmathparse{\x+1} \let\z\pgfmathresult
	        \pgfmathparse{\y+1} \let\t\pgfmathresult
	        \draw[fill=white] (\x,\y) rectangle (\x+1,\y+1);
	        \draw (\x,\y) -- (\z,\t);
	        \draw (\z,\y) -- (\x,\t);
	      }
	      	      \foreach \x/\y in {5/2} {
	        \pgfmathparse{\x+1} \let\z\pgfmathresult
	        \pgfmathparse{\y+1} \let\t\pgfmathresult
	        \draw[fill=white] (\x,\y) rectangle (\x+1,\y+1);
}
        \end{tikzpicture}
      }
      \caption{A right triangle.}
      \label{fig tri}
    \end{figure}
  \end{minipage}
  
\section{The various combinations}

In the sequel of the paper, for each combination of operations, we proceed as follows:
\begin{itemize}
\item[$\bullet$] First, we consider a certain kind of states, computed by applying some constraints on the states of Table \ref{combop}. The enumeration of these states gives the  state complexity of each combination of operations.
\item[$\bullet$] Then, we provide a  Brzozowski witness, often common for all $\circ$ operators, over an alphabet with a cardinality lower than the one described in the literature.
\item[$\bullet$] Finally, we show the accessibility and the pairwise non-equivalence of the states for this witness.
\end{itemize}

\subsection{Double catenation}

In \cite{CLP16}, we give a $3$-letters Brzozowski  witness for the double catenation:
\[{}
W^{(\oo\cdot\oo)\cdot\oo}_{m,n,p}=\left(\chi_m(b,c,-;\{a\}),\chi_n(a,b,c;\emptyset),\chi_p(a,-,b;\{c\})\right).
\]
 This witness is given in Figure \ref{3-letters}.
A $2$-letters witness for the catenation is also given in  \cite{CLP16}. It can be  deduced from the previous one by considering only the two last automata restricted to  letters $a$ and $b$. This witness will be useful in the case of the combined operations $(\oo\circ \oo)\cdot \oo$.

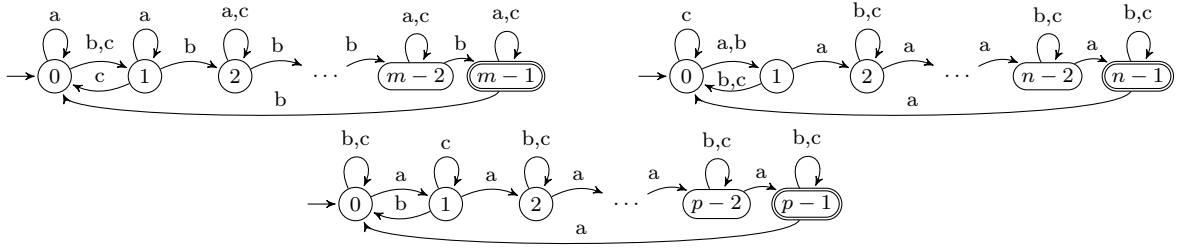
\begin{figure}[htb]
	\centerline{
		\begin{tikzpicture}[node distance=1.2cm, bend angle=25]
			\node[state,initial] (p0) {$0$};
			\node[state] (p1) [right of=p0] {$1$};
			\node[state] (p2) [right of=p1] {$2$};
			\node (etc1) [right of=p2] {$\ldots$};
			\node[state, rounded rectangle] (m-2) [right of=etc1] {${m-2}$};
			\node[state, rounded rectangle, accepting] (m-1) [right of=m-2] {${m-1}$};
			\path[->]
        (p0) edge[bend left] node {b,c} (p1)
        (p1) edge[bend left] node {b} (p2)
        (p2) edge[bend left] node {b} (etc1)
        (etc1) edge[bend left] node {b} (m-2)
        (m-2) edge[bend left] node {b} (m-1)
        (m-1) edge[out=-115, in=-65, looseness=.2] node[above] {b} (p0)
		    (p0) edge[out=115,in=65,loop] node {a} (p0)
		    (p1) edge[out=115,in=65,loop] node {a} (p1)
		    (p2) edge[out=115,in=65,loop] node {a,c} (p2)
		    (m-2) edge[out=115,in=65,loop] node {a,c} (m-2)
		    (m-1) edge[out=115,in=65,loop] node {a,c} (m-1)
        (p1) edge[bend left] node[above] {c} (p0)
			;
			\node (concon) [right of=m-1]{};
			\node[state,initial] (q0) [right of=concon]{$0$};
			\node[state] (q1) [right of=q0] {$1$};
			\node[state] (q2) [right of=q1] {$2$};
			\node (etc2) [right of=q2] {$\ldots$};
			\node[state, rounded rectangle] (n-2) [right of=etc2] {${n-2}$};
			\node[state, rounded rectangle, accepting] (n-1) [right of=n-2] {${n-1}$};
			\path[->]
        (q0) edge[bend left] node {a,b} (q1)
        (q1) edge[bend left] node {a} (q2)
        (q2) edge[bend left] node {a} (etc2)
        (etc2) edge[bend left] node {a} (n-2)
        (n-2) edge[bend left] node {a} (n-1)
        (n-1) edge[out=-115, in=-65, looseness=.2] node[above] {a} (q0)
		    (q0) edge[out=115,in=65,loop] node {c} (q0)
		    (q2) edge[out=115,in=65,loop] node {b,c} (q2)
		    (n-2) edge[out=115,in=65,loop] node {b,c} (n-2)
		    (n-1) edge[out=115,in=65,loop] node {b,c} (n-1)
        (q1) edge[bend left] node[above=-.1cm] {b,c} (q0)
			;
			\node (concon2) [below of=etc1, node distance=1.7cm]{};
			\node[state,initial] (r0) [right of=concon2, node distance=.4cm]{$0$};
			\node[state] (r1) [right of=r0] {$1$};
			\node[state] (r2) [right of=r1] {$2$};
			\node (etc3) [right of=r2] {$\ldots$};
			\node[state, rounded rectangle] (p-2) [right of=etc3] {${p-2}$};
			\node[state, rounded rectangle, accepting] (p-1) [right of=p-2] {${p-1}$};
			\path[->]
        (r0) edge[bend left] node {a} (r1)
        (r1) edge[bend left] node {a} (r2)
        (r2) edge[bend left] node {a} (etc3)
        (etc3) edge[bend left] node {a} (p-2)
        (p-2) edge[bend left] node {a} (p-1)
        (p-1) edge[out=-115, in=-65, looseness=.2] node[above] {a} (r0)
		    (r0) edge[out=115,in=65,loop] node {b,c} (r0)
		    (r1) edge[out=115,in=65,loop] node {c} (r1)
		    (r2) edge[out=115,in=65,loop] node {b,c} (r2)
		    (p-2) edge[out=115,in=65,loop] node {b,c} (p-2)
		    (p-1) edge[out=115,in=65,loop] node {b,c} (p-1)
        (r1) edge[bend left] node[above] {b} (r0)
			;
    \end{tikzpicture}
  }
  \caption{$3$-letters witness for double catenation}
  \label{3-letters}
\end{figure}

\subsection{Combinations of boolean operations}
In this section, we consider the operators $(\oo\circ_{1}\oo)\circ_2\oo$ where $\circ_{1},\circ_{2}\in\{\cup,\cap,\oplus\}$. Notice that, since the operators $\circ_{1}$ and $\circ_{2}$ are commutative, the state complexity of  $\oo\circ_{2}(\oo\circ_{1}\oo)$ is the same as the one of $(\oo\circ_{1}\oo)\circ_{2}\oo$.

The witness we consider, represented in Figure \ref{2-letters bool}, is
\[{}
W^{(\star\circ_{1}\star)\circ_{2}\star}_{m,n,p}=\left({\mathcal X}_m(a,-,-;\{b\}),{\mathcal X}_n(a,b,-;\emptyset),{\mathcal X}_p(b,-,-;\{a\})\right).
\]

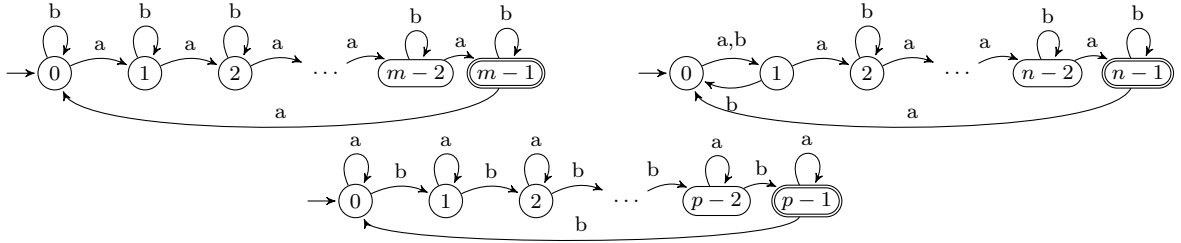
\begin{figure}[htb]
	\centerline{
		\begin{tikzpicture}[node distance=1.2cm, bend angle=25]
			\node[state,initial] (p0) {$0$};
			\node[state] (p1) [right of=p0] {$1$};
			\node[state] (p2) [right of=p1] {$2$};
			\node (etc1) [right of=p2] {$\ldots$};
			\node[state, rounded rectangle] (m-2) [right of=etc1] {${m-2}$};
			\node[state, rounded rectangle, accepting] (m-1) [right of=m-2] {${m-1}$};
			\path[->]
        (p0) edge[bend left] node {a} (p1)
        (p1) edge[bend left] node {a} (p2)
        (p2) edge[bend left] node {a} (etc1)
        (etc1) edge[bend left] node {a} (m-2)
        (m-2) edge[bend left] node {a} (m-1)
        (m-1) edge[out=-115, in=-65, looseness=.33] node[above] {a} (p0)
		    (p0) edge[out=115,in=65,loop] node {b} (p0)
		    (p1) edge[out=115,in=65,loop] node {b} (p1)
		    (p2) edge[out=115,in=65,loop] node {b} (p2)
		    (m-2) edge[out=115,in=65,loop] node {b} (m-2)
		    (m-1) edge[out=115,in=65,loop] node {b} (m-1)
			;
			\node (concon) [right of=m-1]{};
			\node[state,initial] (q0) [right of=concon]{$0$};
			\node[state] (q1) [right of=q0] {$1$};
			\node[state] (q2) [right of=q1] {$2$};
			\node (etc2) [right of=q2] {$\ldots$};
			\node[state, rounded rectangle] (n-2) [right of=etc2] {${n-2}$};
			\node[state, rounded rectangle, accepting] (n-1) [right of=n-2] {${n-1}$};
			\path[->]
        (q0) edge[bend left] node {a,b} (q1)
        (q1) edge[bend left] node {a} (q2)
        (q1) edge[bend left] node {b} (q0)
        (q2) edge[bend left] node {a} (etc2)
        (etc2) edge[bend left] node {a} (n-2)
        (n-2) edge[bend left] node {a} (n-1)
        (n-1) edge[out=-115, in=-65, looseness=.33] node[above] {a} (q0)
		    (q2) edge[out=115,in=65,loop] node {b} (q2)
		    (n-2) edge[out=115,in=65,loop] node {b} (n-2)
		    (n-1) edge[out=115,in=65,loop] node {b} (n-1)
			;
			\node (concon2) [below of=etc1, node distance=1.7cm]{};
			\node[state,initial] (r0) [right of=concon2, node distance=.4cm]{$0$};
			\node[state] (r1) [right of=r0] {$1$};
			\node[state] (r2) [right of=r1] {$2$};
			\node (etc3) [right of=r2] {$\ldots$};
			\node[state, rounded rectangle] (p-2) [right of=etc3] {${p-2}$};
			\node[state, rounded rectangle, accepting] (p-1) [right of=p-2] {${p-1}$};
			\path[->]
        (r0) edge[bend left] node {b} (r1)
        (r1) edge[bend left] node {b} (r2)
        (r2) edge[bend left] node {b} (etc3)
        (etc3) edge[bend left] node {b} (p-2)
        (p-2) edge[bend left] node {b} (p-1)
        (p-1) edge[out=-115, in=-65, looseness=.2] node[above] {b} (r0)
		    (r0) edge[out=115,in=65,loop] node {a} (r0)
		    (r1) edge[out=115,in=65,loop] node {a} (r1)
		    (r2) edge[out=115,in=65,loop] node {a} (r2)
		    (p-2) edge[out=115,in=65,loop] node {a} (p-2)
		    (p-1) edge[out=115,in=65,loop] node {a} (p-1)
			;
    \end{tikzpicture}
  }
  \caption{$2$-letters witness for any combination of two boolean operations}
  \label{2-letters bool}
\end{figure}

According to the construction described in Section \ref{Brzozowski witnesses}, we examine the reachability and the pairwise equivalence of the states of the automaton
\[{}
\mathcal R^{(\star\circ_{1}\star)\circ_{2}\star}_{m,n,p}=\left({\mathcal X}_m(a,-,-;\{b\})\circ_{1}{\mathcal X}_n(a,b,-;\emptyset)\right)\circ_{2}{\mathcal X}_p(b,-,-;\{a\}).
\]
From Table \ref{combop}, the states of $\mathcal R^{(\star\circ_{1}\star)\circ_{2}\star}_{m,n,p}$ are under the form $(i,j,k)$ with $0\leq i\leq m-1$, $0\leq j\leq n-1$, and $0\leq k\leq p-1$. 
We need the following which is straightforward from the construction.
\begin{lemma}\label{lm-permut}
Let $A=(\Sigma, Q_A,i_A,F_A,\cdot_A)$ and $B=(\Sigma, Q_B,i_B,F_B,\cdot_B)$. If $w\in \Sigma^*$ induces a permutation on both $Q_A$ and $Q_B$, then it also induces a permutation on $Q_{A\circ B}$, for any $\circ \in \{\cap, \cup, \oplus\}$.
\end{lemma}
Noticing that any word $w\in \{a,b\}^*$ induces a permutation on the states of $\chi_m(a,-,-,\{b\})$,  $\chi_n(a,b,-,\emptyset)$, and $\chi_p(b,-,-,\{a\})$, we apply Lemma \ref{lm-permut} to show that $w$ induces a permutation on the states of $\mathcal R^{(\star\circ_{1}\star)\circ_{2}\star}_{m,n,p}$.

\begin{lemma}\label{lm-inverse}
Let $A=(\Sigma, Q_A,i_A,F_A,\cdot_A)$. If $w\in \Sigma^*$ induces a permutation $\sigma$ on $Q_A$, then there exists a word $u$ inducing the permutation $\sigma^{-1}$ on $Q_A$. In other words, we have $(s\cdot u)\cdot w=(s\cdot w)\cdot u=s$ for any $s$ of $Q_A$. We denote by $w^{-1}$ such a word $u$ and $w^{-i}=(w^{-1})^i=(w^{i})^{-1}$.
\end{lemma}
\begin{proof}
Since $\sigma$ is a permutation, there exists $N\in \mathbb{N}$ such that $\sigma^{N}=Id$. It suffices to set $N$ as the \emph{lcm} of the sizes of the cycles of $\sigma$. Hence, $u=w^{N-1}$ acts as $\sigma^{N-1}=\sigma^{-1}$ on $Q_A$.
\end{proof}
For any $j\in \mathbb{Z}$, we define $u_{j}=a^{-j}ba^{j}$. The action of $u_{j}$ on the states is given by
	\begin{equation}
		(i,j,k)\cdot u_{j'}=\left\{\begin{array}{ll}(i,j,k+1)&\mbox{if }j\not\in \{\mmod{j'}{n},\mmod{j'+1}{n}\},\\
		(i,j+1,k+1)&\mbox{if }j=\mmod{j'}{n},\\
		(i,j-1,k+1)&\mbox{if }j=\mmod{j'+1}{n}.\end{array}\right.
	\end{equation}
	So we have
	\begin{equation}\begin{array}{rcl}
		(0,0,0)\cdot u_{0}u_{1}\cdots u_{[j-i-2]_{p}}=(0,1,1)\cdot u_{1}u_{2}\cdots u_{[j-i-2]_{p}}&=&(0,2,2)\cdot{}
		 u_{2}u_{3}\cdots u_{[j-i-2]_{p}}\\&=&\cdots=(0,j-i-1,j-i-1).\end{array}
	\end{equation}
	We define for any triplet $(i,j,k)$ the word $w_{ijk}= u_{0}u_{1}\cdots u_{j-i-2}a^{i}u_{j-1}u_{j+1}^{k+i-j}$.
	We have
	\begin{equation}\begin{array}{rcl}
	(0,0,0)\cdot w_{ijk}=(0,j-i-1,j-i-1)\cdot a^{i}u_{j-1}u_{j+1}^{k+i-j}&=&(i,j-1,j-i-1)\cdot
	u_{j-1}u_{j+1}^{k+i-j}\\&=&(i,j,j-i)\cdot u_{j+1}^{k+i-j}.\end{array}
	\end{equation}
	But since $n>2$, we have $j\not\in \{\mmod{j+1}{n},\mmod{j+2}{n}\}$. Hence,
	\begin{equation}\label{wijk}
		(0,0,0)\cdot w_{ijk}=(i,j,k).
	\end{equation}
	The accessibility follows immediately from equation (\ref{wijk}).
	\begin{proposition}\label{o1o2acc}
		All the states of $\mathcal R^{(\star\circ_{1}\star)\circ_{2}\star}$ are accessible.
	\end{proposition}
In the aim to prove the pairwise non equivalence, we need a slightly more general result.
\begin{lemma}\label{0002ijk}
	Let $(i_{s},j_{s},k_{s})$ and $(i_{d},j_{d},k_{d})$ be two states of $\mathcal R^{(\star\circ_{1}\star)\circ_{2}\star}$. There exists a word
	$w$ such that $(i_{s},j_{s},k_{s})\cdot w=(i_{d},j_{d},k_{d})$ .
\end{lemma}
\begin{proof}
 It suffices to set $w=w_{i_{s}j_{s}k_{s}}^{-1}w_{i_{d}j_{d}k_{d}}.$
\end{proof}
For the sake of simplicity, we adopt the following notation: 
	\begin{itemize} 
	\item a word of states $i_{1}i_{2}\cdots i_{t}$ denotes the set $\{i_{1},i_{2},\dots,i_{t}\}$.
	\item if $I=\{i_{1},\cdots,i_{k}\}\subset Q$ is a set of states of an automaton $(\Sigma,Q,i,F,\cdot)$, we denote 
	$\overline {i_{1}\cdots{}
	i_{k}}=Q\setminus I$, 
	\item finally, we use $\_$ instead of $\overline\varepsilon=Q$. 
		\end{itemize}
The final states of $\mathcal R^{(\star\circ_{1}\star)\circ_{2}\star}$ are summarized in the following table:
\begin{equation}\label{compo1o2}
\begin{array}{|c|c|c|c|}
\hline
o_{1}\setminus o_{2}&\cup&\cap&\oplus\\\hline
\cup&\begin{array}{c}
(m-1,\_,\_)\\(\_,n-1,\_)\\(\_,\_,p-1)
\end{array}&
\begin{array}{c}
(m-1,\_,p-1)\\(\_,n-1,p-1)
\end{array}&
\begin{array}{c}\rule[0cm]{0cm}{0.5cm}
(m-1,\_,\overline{p-1})\\\rule[0cm]{0cm}{0.5cm}(\_,n-1,\overline{p-1})\\\rule[0cm]{0cm}{0.5cm}(\overline{m-1},\overline{n-1},p-1)
\end{array}\\\hline
\cap&{}
\begin{array}{c}
(m-1,n-1,\_)\\(\_,\_,p-1)
\end{array}&(m-1,n-1,p-1)&
\begin{array}{c}\rule[0cm]{0cm}{0.5cm}
(m-1,n-1,\overline{p-1})\\\rule[0cm]{0cm}{0.5cm}(\overline{m-1,n-1},p-1)
\end{array}\\\hline
\oplus&\begin{array}{c}\rule[0cm]{0cm}{0.5cm}
(m-1,\overline{n-1},\_)\\\rule[0cm]{0cm}{0.5cm}(\overline{m-1},n-1,\_)\\(\_,\_,p-1)
\end{array}&
\begin{array}{c}\rule[0cm]{0cm}{0.5cm}
(m-1,\overline{n-1},p-1)\\\rule[0cm]{0cm}{0.5cm}(\overline{m-1},n-1,p-1)
\end{array}&
\begin{array}{c}\rule[0cm]{0cm}{0.5cm}
(m-1,\overline{n-1},\overline{p-1})\\\rule[0cm]{0cm}{0.5cm}(\overline{m-1},n-1,\overline{p-1})\\
\rule[0cm]{0cm}{0.5cm}(\overline{m-1},\overline{n-1},p-1)\\(m-1,n-1,p-1)
\end{array}\\\hline
\end{array}
\end{equation}
We warn the reader than for  $\circ_{1}=\oplus$ and $\circ_{2}=\cap$, $\overline{m-1,n-1}$ means we consider the complementary set of $\{(m-1,n-1)\}$
with respect to the states of  $\mathcal R^{\star\circ_{1}\star}$.
\begin{proposition}
	The states of $\mathcal R^{(\star\circ_{1}\star)\circ_{2}\star}$ are pairwise non equivalent.
\end{proposition}
\begin{proof}
	Since each word acts as a permutation and all the states are accessible (Proposition \ref{o1o2acc}), it suffices to prove that
	for any $(i,j,k)\neq (0,0,0)$ the states $(0,0,0)$ and $(i,j,k)$ are non equivalent. Indeed, assuming this fact, if $(i',j',k')\neq (i,j,k)$
	then $(i',j',k')\cdot w_{i'j'k'}^{-1}=(0,0,0)$ and  the non equivalence of $(i',j',k')$ and $(i,j,k)$  comes from the non equivalence of $(0,0,0)$ and $(i,j,k)\cdot w_{i'j'k'}^{-1}$; this
	last state is different of $(0,0,0)$ because $w_{i'j'k'}^{-1}$ acts as a permutation and $(0,0,0)$ already has a preimage.\\ \\
	A consequence of Lemma \ref{0002ijk} is that it is not necessary to find a word $w$ separating $(0,0,0)$ and $(i,j,k)$ but only to investigate
	the preimage of $(m-1,n-1,p-1)$ by $w$. We choose the preimage $s$ of $(m-1,n-1,p-1)$ in  a set depending on the operations $\circ_{1}$ and $\circ_{2}$. This set is described in the following table:
\begin{equation}
\begin{array}{|c|c|c|c|}
\hline
o_{1}\setminus o_{2}&\cup&\cap&\oplus\\\hline
\cup&\begin{array}{rl}   \rule[0cm]{0cm}{0.5cm}
\mbox{if }k\neq 0,&s\in(\overline{0i},\overline{0j},0)\ \ \ \ \ \ \ \ \ \ \ \\\hline
\rule[0cm]{0cm}{0.5cm}\mbox{if }j\neq 0,&s\in(\overline{0i},0,\overline{0k})\\\hline
\rule[0cm]{0cm}{0.5cm}\mbox{if }i\neq 0,&s\in(0,\overline{0j},\overline{0k})
\end{array}&
\begin{array}{ll}
\rule[0cm]{0cm}{0.5cm}\mbox{if }k\neq 0,&s\in(0,0,0)\\\hline
\rule[0cm]{0cm}{0.5cm}\mbox{if }j\neq 0,&s\in(\overline{0i},0,0)\\\hline
\rule[0cm]{0cm}{0.5cm}\mbox{if }i\neq 0,&s\in(0,\overline{0j},0)
\end{array}&
	\begin{array}{rl}
\rule[0cm]{0cm}{0.5cm}\mbox{if }k\neq 0,&s\in(\overline{0i},\overline{0j},0)\ \ \ \ \ \ \ \ \ \ \ \ \ \ \\\hline
\rule[0cm]{0cm}{0.5cm}\mbox{if }j\neq 0,&s\in(\overline{i},0,\overline{0k})\\\hline
\rule[0cm]{0cm}{0.5cm}\mbox{if }i\neq 0,&s\in(0,\overline{j},\overline{0k})
\end{array}\\\hline
\cap{}&
\begin{array}{rl}
\rule[0cm]{0cm}{0.5cm}\mbox{if }k\neq 0&,s\in(\overline{i},\overline{j},0)\\\hline
\rule[0cm]{0cm}{0.5cm}\mbox{if }i\neq 0&\mbox{ or }j\neq 0,s\in(0,0,\overline{0k})
\end{array}
&s=(0,0,0)&
\begin{array}{rl}
\rule[0cm]{0cm}{0.5cm}\mbox{if }k\neq 0&,\ s\in(\overline{0i},\overline{0j},0)\\\hline
\rule[0cm]{0cm}{0.5cm}\mbox{if }i\neq 0&\mbox{ or }j\neq 0,s\in(0,0,\overline{0k})\ \ \ 
\end{array}\\\hline
\oplus&{}
\begin{array}{rl}
\rule[0cm]{0cm}{0.5cm}\mbox{if }k\neq 0,&s\in(\overline{i},\overline{j},0)\ \ \ \ \ \ \ \ \ \ \ \ \ \ \\\hline
\rule[0cm]{0cm}{0.5cm}\mbox{if }j\neq 0,&s\in(\overline{0i},0,\overline{0k})\\\hline
\rule[0cm]{0cm}{0.5cm}\mbox{if }i\neq 0&s\in(0,\overline{0j},\overline{0k})
\end{array}&
\begin{array}{rl}
\rule[0cm]{0cm}{0.5cm}\mbox{if }k\neq 0,&s\in(0,\overline{0},0)\\\hline
\rule[0cm]{0cm}{0.5cm}\mbox{if }j\neq 0,&s\in(\overline{0i},0,0)\\\hline
\rule[0cm]{0cm}{0.5cm}\mbox{if }i\neq 0,&s\in(0,\overline{0j},0)
\end{array}&
\begin{array}{rl}
\rule[0cm]{0cm}{0.5cm}\mbox{if }k\neq 0,&s\in(\overline{0i},\overline{0j},0)\ \ \ \ \ \ \ \ \ \ \ \ \ \ \\\hline
\rule[0cm]{0cm}{0.5cm}\mbox{if }j\neq 0,&s\in(\overline{0i},0,\overline{0k})\\\hline
\rule[0cm]{0cm}{0.5cm}\mbox{if }i\neq 0&s\in(0,\overline{0j},\overline{0k})
\end{array}\\\hline
\end{array}
\end{equation}
For such a state $s$, from Lemma \ref{0002ijk}, there exists word $w_{s}$ such that $s\cdot w_{s}=(m-1,n-1,p-1)$. Comparing to (\ref{compo1o2}), we
check that $(0,0,0)\cdot w_{s}$ is final while $(i,j,k)\cdot w_{s}$ is not final. For instance, consider the case $\circ_{1}=\circ_{2}=\cup$ and $k\neq 0$, 
we have $(0,0,0)\cdot w_{s}\in(\overline{m-1},\overline{n-1},p-1)$ which is final and $(i,j,k)\cdot w_{s}\in(\overline{m-1},\overline{n-1},\overline{p-1})$
which is not final.	
\end{proof}

The following theorem summarizes the results of the section.
\begin{theorem}
	The state complexity of any combination of two boolean operations is $mnp$ and the bound is reached for two-letters witnesses.
\end{theorem}


\subsection{$A\cdot(B\circ C)$}
We consider the witness 
$$W_{m,n,p}^{\star\cdot(\star\circ\star)}=
(\mathcal X_{m}(a,b,-;\{c\}),\mathcal X_{n}(b,a,c;\emptyset),\mathcal X_{p}(b,-,c;\{a\}))$$ 
for each $m,n,p\geq 3$ and $\circ\in\{\cap,\cup,\oplus\}$ (see Figure \ref{WA.(BoC)}).

\begin{figure}[htb]
	\centerline{
		\begin{tikzpicture}[node distance=1.2cm, bend angle=25]
			\node[state,initial] (p0) {$0$};
			\node[state] (p1) [right of=p0] {$1$};
			\node[state] (p2) [right of=p1] {$2$};
			\node (etc1) [right of=p2] {$\ldots$};
			\node[state, rounded rectangle] (m-2) [right of=etc1] {$m-2$};
			\node[state, rounded rectangle, accepting] (m-1) [right of=m-2] {$m-1$};
			\path[->]
        (p0) edge[bend left] node {a,b} (p1)
        (p1) edge[bend left] node {a} (p2)
        (p2) edge[bend left] node {a} (etc1)
        (etc1) edge[bend left] node {a} (m-2)
        (m-2) edge[bend left] node {a} (m-1)
        (m-1) edge[out=-115, in=-65, looseness=.2] node[above] {a} (p0)
		    (p0) edge[out=115,in=65,loop] node {c} (p0)
		    (p1) edge[out=115,in=65,loop] node {c} (p1)
		    (p2) edge[out=115,in=65,loop] node {b,c} (p2)
		    (m-2) edge[out=115,in=65,loop] node {b,c} (m-2)
		    (m-1) edge[out=115,in=65,loop] node {b,c} (m-1)
			(p1) edge[bend left] node [swap] {b} (p0)
			;
			\node (concon) [right of=m-1]{};
			\node[state,initial] (q0) [right of=concon]{$0$};
			\node[state] (q1) [right of=q0] {$1$};
			\node[state] (q2) [right of=q1] {$2$};
			\node (etc2) [right of=q2] {$\ldots$};
			\node[state, rounded rectangle] (n-2) [right of=etc2] {$n-2$};
			\node[state, rounded rectangle, accepting] (n-1) [right of=n-2] {$n-1$};
			\path[->]
        (q0) edge[bend left] node {a,b} (q1)
        (q1) edge[bend left] node {b} (q2)
        (q2) edge[bend left] node {b} (etc2)
        (etc2) edge[bend left] node {b} (n-2)
        (n-2) edge[bend left] node {b} (n-1)
        (n-1) edge[out=-115, in=-65, looseness=.2] node[above] {b} (q0)
		    (q0) edge[out=115,in=65,loop] node {c} (q0)
		    (q2) edge[out=115,in=65,loop] node {a,c} (q2)
		    (n-2) edge[out=115,in=65,loop] node {a,c} (n-2)
		    (n-1) edge[out=115,in=65,loop] node {a,c} (n-1)
             (q1) edge[bend left] node[swap] {a,c} (q0)
			;
			\node (concon2) [below of=etc1, node distance=1.7cm]{};
			\node[state,initial] (r0) [right of=concon2, node distance=.4cm]{$0$};
			\node[state] (r1) [right of=r0] {$1$};
			\node[state] (r2) [right of=r1] {$2$};
			\node (etc3) [right of=r2] {$\ldots$};
			\node[state, rounded rectangle] (p-2) [right of=etc3] {$p-2$};
			\node[state, rounded rectangle, accepting] (p-1) [right of=p-2] {$p-1$};
			\path[->]
        (r0) edge[bend left] node {b} (r1)
        (r1) edge[bend left] node {b} (r2)
        (r2) edge[bend left] node {b} (etc3)
        (etc3) edge[bend left] node {b} (p-2)
        (p-2) edge[bend left] node {b} (p-1)
        (p-1) edge[out=-115, in=-65, looseness=.2] node[above] {b} (r0)
		    (r0) edge[out=115,in=65,loop] node {a,c} (r0)
		    (r1) edge[out=115,in=65,loop] node {a} (r1)
		    (r2) edge[out=115,in=65,loop] node {a,c} (r2)
		    (p-2) edge[out=115,in=65,loop] node {a,c} (p-2)
		    (p-1) edge[out=115,in=65,loop] node {a,c} (p-1)
        (r1) edge[bend left] node[above] {c} (r0)
			;
    \end{tikzpicture}
  }
  \caption{$3$-letters witness for $A\cdot (B\circ C)$}
  \label{WA.(BoC)}
\end{figure}
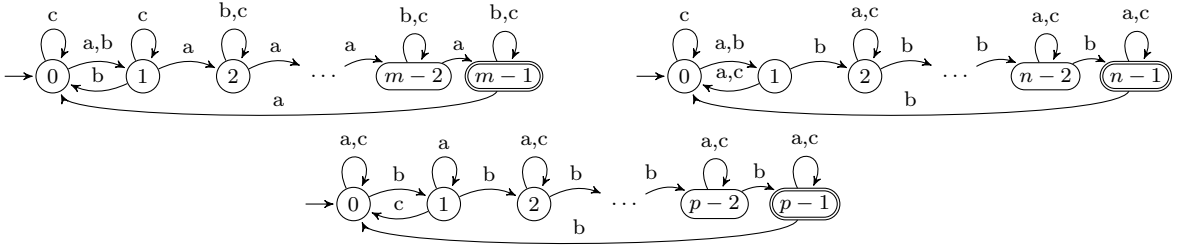

According to the constructions described in Section \ref{Brzozowski witnesses}, we 
define, for each $m,n,p\geq 3$ the automaton
 $$\mathcal 
{R}^{\star\cdot(\star\circ\star)}_{m,n,p}=\mathcal X_{m}(a,b,-;\{c\})\cdot\left(\mathcal X_{n}(b,a,c;\emptyset)\circ
\mathcal X_{p}(b,-,c;\{a\})\right).$$
From Table \ref{combop}, the  accessible states of $\mathcal 
{R}^{\star\cdot(\star\circ\star)}_{m,n,p}$ are indexed by pairs $(i,T)$ where $0\leq i\leq m-1$ and $T\subset [0,n-1]
\times [0,p-1]$.
  The transitions 
are described as follows: for each pair $(i,T)$ and each symbol $\sigma$,{}
$$(i,T)\cdot {\sigma}=\left\{\begin{array}{ll}
(i\cdot \sigma,\{({j}\cdot \sigma,{k}\cdot\sigma):({j},{k})\in T\})&\mbox{ if }{i}\cdot \sigma\neq {m-1}\\
(i\cdot \sigma,\{(j\cdot \sigma,k\cdot\sigma):(j,k)\in T\}\cup\{(0,0)\})&\mbox{ if }i\cdot \sigma={m-1}\\
\end{array}\right.$$
It is easy to see that only 
the states $(i,T)$ satisfying $i=m-1\Rightarrow (0,0)\in T$ are accessible.
We set 
$$\begin{array}{rcl}
Acc^{\star\cdot(\star\circ\star)}_{m,n,p}&=&\{(i,T):0\leq i<m-1, T\subset  Q_n \times Q_p\}\\&&
\cup \{({m-1},\{(0,0)\}\cup T): T\subset Q_n \times Q_p\}\end{array}$$
Notice that the set $Acc^{\star\cdot(\star\circ\star)}_{m,n,p}$  does not depend on $\circ$.
\begin{lemma}\label{LRacc}
For any  state $s=(i,\{(j,k)\})$ of ${\cal R}^{\star\cdot(\star\circ\star)}_{m,n,p}$ with $i\neq m-1$,
 there exists a word $w\in \{a,b\}^*$ such that $(m-1,\{(0,0)\})\cdot w=s$.
\end{lemma}
\begin{proof}

	We have to consider three cases:
	\begin{enumerate}
		\item If $j>2$ then we have
		\[\displaystyle(m-1,\{(0,0)\})\xrightarrow{a}(0,\{(1,0)\})
		\xrightarrow{(ab)^{\mmod{k-j+1}{p}}}(0,\{(1,k-j+1)\}\xrightarrow{b^{j-2}}(\mmod{j}{2},\{(j-1,k-1)\}).\]
		Since $j-1>1$, one obtains
		\[{}
		(\mmod{j}{2},\{(j-1,k-1)\})\xrightarrow{a^{[j+1]_2}} (1,\{(j-1,k-1)\})\xrightarrow{b}{}
		(0,\{(j,k)\}\xrightarrow{a^{i}} (i,\{(j,k)\}).
		\]
		In conclusion, the word $$w=a(ab)^{\mmod{k-j+1}{p}}b^{j-2}a^{\mmod{j+1}{2}}ba^{i}\in\{a,b\}^{*}$$ is such that 
		$(m-1,\{(0,0)\}).w=(i,\{(j,k)\})$. This proves the lemma.
		\item If $j=2$ then we have
		\[\displaystyle(m-1,\{(0,0)\})\xrightarrow{a}(0,\{(1,0)\})
		\xrightarrow{(ab)^{\mmod{k-1}{p}}}(0,\{(1,k-1)\}\xrightarrow{b^{np-2}}(\mmod{np}{2},\{(n-1,k-3)\}).\]
		Hence,
		\[{}
		(\mmod{np}{2},\{(n-1,k-3)\})\xrightarrow{a^{\mmod{np+1}{2}}} (1,\{(n-1,k-3)\})\xrightarrow{b^{3}}{}
		(0,\{(2,k)\}\xrightarrow{a^{i}}(i,\{(2,k)\}).
		\]
		In conclusion, the word $$w=a(ab)^{\mmod{k-1}{p}}b^{np-2}a^{\mmod{np+1}{2}}b^{3}a^{i}\in\{a,b\}^{*}$$ is such that 
		$(m-1,\{(0,0)\})\cdot w=(i,\{(j,k)\})$. This proves the lemma.
		\item If $j<2$ then set $\gamma_{i}=i$ if $i>1$ and $\gamma_i=\mmod{i+j+1}{2}$ if $i\leq 1$. One has
		\[{}
		(\gamma_{i},\{(n-1,k-j-1)\})\xrightarrow{b^{j+1}} (i,\{(j,k)\}).
		\]
		Hence, as $\gamma_{i}\neq m-1$, there exists a word $v$ such that $(m-1,\{(0,0)\})\xrightarrow{v}(\gamma_{i},n-1,k-j-1)$  due to one of the previous cases.
		So one obtains 
		$(m-1,\{(0,0)\})\cdot vb^{j+1}=(i,\{(j,k)\})$ as expected.
	\end{enumerate}
\end{proof}
\begin{proposition}\label{RCacc}
	For any boolean operation $\circ$, all the states of $Acc^{\star\cdot(\star\circ\star)}_{m,n,p}$ are accessible in 
	$\mathcal 
{R}^{\star\cdot(\star\circ\star)}_{m,n,p}$.
\end{proposition}
\begin{proof}
We prove by induction on $|T|$ that each state $(i,T)$ is accessible.
	First, observe that all the states  $(i,\emptyset)$ are 
	reachable from $(0,\emptyset)$ reading $a^{i}$. Now consider a state 
	$(i,T)$ with $T\neq\emptyset$.
	
\begin{enumerate}
\item\label{item1} Suppose that $i=m-1$ then the states $(m-1,T)$ is reachable by $a$ from 
	$(m-2,a\cdot(T\setminus \{(0,0)\})$ which is accessible by 
	induction.
	
\item Suppose now $i<m-1$ and let $(j,k)\in T$. Since $i\neq m-1$, by Lemma 
	\ref{LRacc}, there exists a word $w\in \{a,b\}^{*}$ such that $(m-1,\{(0,0)\})\cdot w=(i,\{(j,k)\})$.
	Observe that, from the definition of the automata, the letter $a$ and $b$ encode permutations of the states (no contraction is involved).
	It follows that $w\cdot T$ has the same number of elements as $T$ and so the state $(i,T)$ is accessible by $w$ from $(m-1,w\cdot T)$ which is accessible from (\ref{item1}). 
	\end{enumerate}
\end{proof}

We remark  that the action of the letter $c$ is not used to prove the accessibility of the states. Nevertheless, this letter is needed  to separate the states. The following lemma highlights a property of the action of $c$ which is  central in the study of the separability. Its proof is straightforward from the definition of $\mathcal R_{m,n,p}^{\star\cdot(\star\circ\star)}$.

\begin{lemma}\label{lc1}
Let $(i,T)$ be a state in ${\mathcal Acc}_{m,n,p}^{\star\cdot(\star\circ\star)}$. We have $(i,T)\cdot c=(i,T')$ with 
$$T'\subset (Q_n\setminus \{1\})\times (Q_p\setminus \{1\})$$
\end{lemma}

Let us consider first the case where $\circ=\cap$.
Notice that a state $(i,T)$ of ${\mathcal Acc}^{\star\cdot(\star\cap\star)}_{m,n,p}$ is final if and only if $(n-1,p-1)\in T$. 
\begin{proposition}
	The states belonging to ${\mathcal Acc}^{\star\cdot(\star\cap\star)}_{m,n,p}$  are pairwise nonequivalent in 
	$\mathcal{R}^{\star\cdot(\star\cap\star)}_{m,n,p}$.
\end{proposition}
\begin{proof}
	Let $s=(i,T)$ and $s'=({i'},T')$ be two distinct states. Without loss of 
	generality we assume $i'\leq i$ (otherwise we permute the role of 
	the states) and we construct a word $w_{s,s'}$ sending one of the state on a final state and the other on a non final state.
	 We consider several cases as follows
	\begin{enumerate}
		\item Suppose  $i'<i<m-1$. 
		We have
		$$s\xrightarrow{a^{m-i-2}c}(m-2,T_{2})\mbox{ and }s'\xrightarrow{a^{m-i-2}c}(m-2+i'-i,T'_{2})$$
		with, from Lemma \ref{lc1},  $({1},{0})\not\in T_{2}\cup T'_{2}$.\\
		Hence,
		\[{}
		(m-2,T_{2})\xrightarrow{a}(m-1,T_{3})
		\xrightarrow{b^{np-1}} (m-1,T_{4})
		\]
		with $({0},{0})\in T_{3}$ and $({n-1},{p-1})\in T_{4}$,
		and
		\[{}
		(m-2+i'-i,T'_{2})\xrightarrow{a}(m-1+i'-i,T'_{3})
		\xrightarrow{b^{np-1}}(i'_{4},T'_{4})
		\]
		with $({0},{0})\not\in T'_{3}$ because both $m-1+i'-i\neq m-1$ and  $(1,0)\not\in T'_{2}$. 
		Furthermore as $m-1$ is never reached from $m-1+i'-i$ reading $b^{np-1}$,  we have 
		$(n-1,p-1)\not\in T'_{4}$.
		Setting $w_{s,s'}=a^{m-i-2}cab^{np-1}$, we have $s\cdot w_{s,s'}=(m-1,T_{4})$ which is final and
		$s'\cdot w_{s,s'}=(i'_{4},T'_{4})$ which is not final. So, $s$ and $s'$ are not equivalent.
	\item If $i'<i=m-1$ then reading $a$  sends $s$ to a state 
		$s_{1}=(0,T_{1})$ and $s'$ to a state $s'_{1}=({i'+1},T'_{1})$. 
		If $i'+1\neq m-1$,then  we set $w_{s,s'}=aw_{s'_{1},s_{1}}$ where $w_{s'_{1},s_{1}}$ is computed from the previous case.
		  If $i'+1=m-1$ then we read another $a$ and this 
		sends $s_{1}$ to a state $s_{2}=({1},T_{2})$ and $s'_{1}$ to a 
		state $s'_{2}=({0},T'_{2})$. As ${\mathcal X} _m$ has at least $3$ states, $m-1\neq 1$.  So 
		$w_{s,s'}=a^{2}w_{s_{2},s'_{2}}$ where $w_{s_{2},s'_{2}}$ is the word computed in the previous case.
	\item If $i=i'$ then $T\neq T'$. Without loss of generality we 
	assume that there exists $({j},{k})\in T\setminus T'$. 
	
	Let us recall the Kronecker delta $\delta_{i,j}=\left\{\begin{array}{ll}0&\text{ if } i\neq j\\1&\text{ if }i=j\end{array}\right.$
	We have :
	$$ s\xrightarrow{a^{\delta_{m-1,i}}}(i_1,T_1)\xrightarrow {b^{n-j_1}}(i_2,T_2)\text{ and }s'\xrightarrow{a^{\delta_{m-1,i}}}(i_1,T'_1)\xrightarrow {b^{n-j_1}}(i_2,T'_2)$$

	where $i_1,i_2\neq m-1$,  $(j_1,k_1)=(j,k)\cdot a^{\delta_{m-1,i}} \in T_1\setminus T'_1$ and $(0,k_2)=(j_1,k_1)\cdot b^{n-j_1}\in T_2\setminus T'_2 $.
	
	$$(i_2,T_2)\xrightarrow{(ba)^{p-k_2-1}}(i_3,T_{3})\xrightarrow{(aa)^{\delta_{i_3,m-2}}}(i_4,T_4)\mbox{ and }(i_2,T'_2)\xrightarrow{(ba)^{p-k_2-1}}(i_3,T'_{3})\xrightarrow{(aa)^{\delta_{i_3,m-2}}}(i_4,T'_4)$$

      where $i_4\neq m-2$ and $(0,p-1)=(0,k_2)\cdot (ba)^{p-k_2-1}\in T_4\setminus T'_4$.

	$$(i_4,T_4)\xrightarrow{ba}(i_5,T_{5})\xrightarrow{b^{np-1}}(i_6,T_6)\mbox{ and }(i_4,T'_4)\xrightarrow{ba}(i_5,T'_{5})\xrightarrow{b^{np-1}}(i_6,T'_6)$$

      where $i_5\neq m-1$ and $(0,0)\in T_5\setminus T'_5$ and finally $(n-1,p-1)\in T_6\setminus T'_6$.
	
	
	Setting $w_{s,s'}=a^{\delta_{m-1,i}}b^{n-j_1}(ba)^{p-k_2-1}(aa)^{\delta_{i_3,m-2}}bab^{np-1}$ we obtain that $s.w_{s,s'}$ is final while $s'.w_{s,s'}$ is not final.
	In other words, $s$ 
	and $s'$ are nonequivalent.
	\end{enumerate}
\end{proof}

Now, we consider the case where $\circ=\cup$. 
The final states of $\mathcal{R}^{\star\cdot(\star\cup\star)}_{m,n,p}$ are the pairs 
$(p_{i},T)$ such that 
$$T\cap\left(\{{n-1}\}\times Q_{p}\cup Q_{n}\times \{{p-1}\}\right)\neq \emptyset.$$

 We say that a set $T$ is \emph{saturated} if $({j},{k}), ({j'},{k'})\in T$  
 implies $({j},{k'})\in T$. 

 \begin{equation}\label{sat}
	 \sat(T)=\{{j}:(j,k)\in T\}\times \{k:(j,k)\in T\}.
 \end{equation}

\begin{lemma}
In ${\mathcal Acc}^{\star\cdot(\star\cup\star)}_{m,n,p}$, any state $(i,T)$ is equivalent to  $(i,\sat(T))$. 
\end{lemma}

\begin{proof}
	Suppose that there exists a word $w$ such that $({i},\sat(T))\cdot w$ is final and $({i},T)\cdot w$ is not final. 
	Then, there are two couples $(j,k)$ and $(j',k')$ in $T$ with $(j,k')\in \sat(T)\setminus T$ and $$({j},{k'})\cdot w\in (\{{n-1}\}\times Q_{p})\cup (Q_{n}\times \{{p-1}\}).$$
	This means that either ${j}\cdot w={n-1}$ or ${k'}\cdot w={p-1}$. But since $({j},{k})\cdot w, ({j'},{k'})\cdot w\in T\cdot w$
	 we have  
	$$T\cdot w\cap\left(\{{n-1}\}\times Q_{p}\cup Q_{n}\times \{{p-1}\}\right)\neq \emptyset$${}
	  and this is not possible because $({i},T)\cdot w=({i}\cdot w,T\cdot w)$ which is not final.
\end{proof}
\begin{lemma}
Let $(i,T)$ and $(i,T')$ be two  states of 
${\mathcal Acc}^{\star\cdot(\star\cup\star)}_{m,n,p}$ such that $\sat(T)=\sat(T')$. Then $(i,T)$ and $(i,T')$ are equivalent.
\end{lemma}
From now, we will only consider  the set of saturated states defined as follows : 
  \[{}
 Sat=\{(i,\sat(T))\in {\mathcal Acc}_{m,n,p}^{\star\cdot(\star\cup\star)}\} \]

%

 We define $\split(T)=(\{j:(j,k)\in T\},\{k:(j,k)\in T\})$.
 For any $s=(i,T)\in Sat$ we define $L(s)=S_{1}$ and $R(s)=S_{2}$ if $\split(T)=(S_{1},S_{2})$.
 With this notation, a state $s$ is final if and only if ${n-1}\in L(s)$ or ${p-1}\in R(s)$. 
Notice that Lemma \ref{lc1} can be restated as follows.
\begin{lemma}\label{lc2}
	Let $s\in {\mathcal Acc}^{\star\cdot(\star\circ\star)}_{m,n,p}$. Then $s\cdot c=s'$ with ${1}\not\in L(s')$ and ${1}\not\in{}
	R(s')$.
\end{lemma}
Now we have defined all the material we need to prove the pairwise non equivalence of the states of $Sat$.

\begin{proposition}
	The states belonging to $Sat$ are pairwise non equivalent in $\mathcal R^{\star\cdot(\star\cup\star)}_{m,n,p}$.
\end{proposition}
\begin{proof}
	Let $s=(i,T)$ and 
	$s'=({i'},T')$ be two distinct states of  $Sat$ . 
Without loss of generality we assume $i'\leq i$. First suppose $i'<i$. We have to consider 
the following cases.
\begin{itemize}
\item If $i'<i<m-1$ then 
we set
\[{}
s\xrightarrow{a^{m-i+1}(ba)^{i-i'-1}}s_{1}=(1,T_{1})\xrightarrow{a^{m-3}c}s_{2}
=({m-2},T_{2})
\]
\[{}
s'\xrightarrow{a^{m-i+1}(ba)^{i-i'-1}}s'_{1}=(0,T'_{1})\xrightarrow{a^{m-3}c}s'_{2}
=({m-3},T'_{2})
\]
Using Lemma \ref{lc2} we observe that ${1}\not\in L(s'_{2})$ and ${1}\not\in R(s'_{2})$.
Setting
\[{}
s_{2}\xrightarrow{a}s_{3}=(m-1,T_{3})\xrightarrow{aba}s_{4}=(2,T_{4})\xrightarrow{b^{np-3}}s_{5}=(2,T_{5}),
\]
We observe that 
${0}\in L(s_3)$, 
 ${2}\in L(s_{4})$, and  then 
  ${n-1}\in L(s_{5})$. In other words $s_{5}$ 
   is a final state.
In the other hand, we set
\[{}
s'_{2}\xrightarrow{a}s'_{3}=(m-2,T'_{3})\xrightarrow{aba}s'_{4}=(0,T'_{4})\xrightarrow{b^{np-3}}s'_{5}=(\epsilon,T'_{5}),
\]
with $\epsilon\in\{0,1\}$.
  We observe that  
${0}\not\in L(s'_3)$ and ${1}\not\in R(s'_{3})$, 
  ${2}\not\in L(s'_{4})$ and ${2}\not\in R(s'_{4})$, and finally 
  ${n-1}\not\in L(s'_5)$ and ${p-1}\not\in R(s'_{5})$.
  In other words $s'_{5}$ 
   is not a final state. Setting $w_{s,s'}=a^{m-i+1}(ba)^{i-i'-1}a^{m-3}caabab^{np-3}$, we obtain that
   $s\cdot w_{s,s'}$ is final but not $s'\cdot w_{s,s'}$. This proves that $s$ and $s'$ are not equivalent.
  \item If $i'<i=m-1$ then, by reading $a$ or $aa$, we recover the case 
where $i'<i<m-1$. 
\end{itemize}

Suppose now that $i=i'$ and $L(s)\neq L(s')$. Without loss of generality we consider ${j}\in L(s)\setminus L(s')$.
 We have  two cases to consider:
 \begin{itemize}
	 \item If $j>1$ then we set
	 \[{}
	 s\xrightarrow{a^{m-i}}s_{1}=(0,T_{1})\xrightarrow{(ab)^{n-1-j}}s_{2}=(0,T_{2})
	 \]
	 and
	 \[{}
	 s'\xrightarrow{a^{m-i}}s'_{1}=(0,T'_{1})\xrightarrow {(ab)^{n-1-j}}s'_{2}=(0,T'_{2})
	 \]
	 We observe that $j\in L(s_{1})\setminus L(s'_{1})$ and $n-1\in L(s_{2})\setminus L(s'_{2})$. We set 
	  \[{}
	  s_{2}\xrightarrow{cb}s_{3}=(1,T_{3})\xrightarrow{(ba)^{\mmod{-n-2}{p}}}s_{4}=(1,T_{4})\xrightarrow{b^{n-1}}s_5=(\epsilon,T_5)
	  \]
	  and
	 \[{}
	  s'_{2}\xrightarrow{cb}s'_{3}=(1,T'_{3})\xrightarrow{(ba)^{\mmod{-n-2}{p}}}s'_{4}=(1,T'_{4})\xrightarrow{b^{n-1}}s'_5=(\epsilon,T'_5)
	  \]
	  with $\epsilon\in\{0,1\}$. We have $0\in L(s_{3})\setminus L(s'_{3})$ and $2\not\in R(s'_{3})$. Hence, $0\in L(s_4)\setminus L(s'_4)$ and $p-n\not\in R(s'_4)$. Finally, 
	  $n-1\in L(s_{5})\setminus L(s'_{5})$ and $p-1\not\in  R(s'_{5})$. In conclusion,
	   if we set $w_{s,s'}=a^{m-i}(ab)^{n-1-j}cb(ba)^{\mmod{-n-2}{p}}b^{n-1}$ then the state $s_{5}=s\cdot w_{s,s'}$ is final while $s'_{5}=s'\cdot w_{s,s'}$ is not.
	   Consequently, $s$ and $s'$ are not equivalent.
	  %
	%
	%
	%
	 \item If $j\leq 1$ then we act by $b$ or $b^{2}$ in the aim to send 
	 $s$ and $s'$ respectively to $s_{1}=(i_{1},T_{1})$ and $s'_{1}=(i_{1},T'_{1})$ with ${2}\in L(s_{1})\setminus L(s'_{1})$.
	  So we find the result by applying the previous point.
 \end{itemize}

 Now we suppose $i=i'$ and $R(s)\neq R(s')$. Without loss of generality we assume that there exists $k\in R(s)\setminus R(s')$.
  We have to consider two cases:
 \begin{itemize}
	 \item If $k>1$ then we set 
	 	 \[{}
	 s\xrightarrow{a^{m-i}}s_{1}=(0,T_{1})\xrightarrow{c}s_{2}=(0,T_{2})
	 \]
	 and
	 \[{}
	 s'\xrightarrow{a^{m-i}}s'_{1}=(0,T'_{1})\xrightarrow {c}s'_{2}=(0,T'_{2})
	 \]
	 We observe that $k\in R(s_{1})\setminus R(s'_{1})$, $1\not\in L(s'_{2})$ and $k\in R(s_2)\setminus R(s'_2)$. We set 
	  \[{}
	  s_{2}\xrightarrow{(ab)^{\mmod{-n-k+1}{p}}}s_{3}=(0,T_{3})\xrightarrow{(b)^{n-2}}s_{4}=(\epsilon,T_{4})
	  \]
	  and
	 \[{}
	  s'_{2}\xrightarrow{(ab)^{\mmod{-n-k+1}{p}}}s'_{3}=(0,T'_{3})\xrightarrow{(b)^{n-2}}s'_{4}=(\epsilon,T'_{4})
	  \]
	  with $\epsilon\in\{0,1\}$. We have $1\not\in L(s'_{3})$ and $p-n+1\in R(s_3)\setminus R(s'_{3})$.  Finally, 
	  $n-1\not\in L(s'_{4})$ and $p-1\in  R(s_{4})\setminus R(s'_4)$. In conclusion,
	   if we set $w_{s,s'}=a^{m-i}c(ab)^{\mmod{-n-k+1}{p}}b^{n-2}$ then the state $s_{4}=s\cdot w_{s,s'}$ is final while $s'_{4}=s'\cdot w_{s,s'}$ is not.
	   Consequently, $s$ and $s'$ are not equivalent.

 	 \item If $k\leq 1$ then we act by $b$ or $b^{2}$ in the aim to send 
	 $s$ and $s'$ respectively to $s_{1}=(i_{1},T_{1})$ and $s'_{1}=(i_{1},T'_{1})$ with ${2}\in R(s_{1})\setminus R(s'_{1})$.
	  So we find the result by applying the previous point.
  \end{itemize}
\end{proof}
The following theorem summarizes the results contained in this section.
\begin{theorem}
	The state complexity  $\mathrm {sc}_{*\cdot (*\cap *)}(m,n,p)$ is $(m-1)2^{np}+2^{np-1}$. The state complexity  $\mathrm {sc}_{*\cdot (*\cup *)}(m,n,p)$ is $(m-1)2^{n+p}+2^{n+p-2}$. In  both cases, the bound is reached by the three-letters witness $W_{m,n,p}^{\star\cdot(\star\circ\star)}$.
\end{theorem}

\subsection*{The symmetric difference case}\label{sect:xor}
Unfortunatly, the family $W_{m,n,p}$ fails for the combination of catenation with boolean xor operator. We prove it by studying the case $m=n=3$ and $p=4$  using tableaux described in Section \ref{preliminaries}. \\

A final state of the catenation combined with the xor has at least one marked cell on the last line or row but not both.\\
Let us show that the two final states represented by  $t=(i,  \begin{tikzpicture}[scale=0.25] 
	      \foreach \x in {1,...,4} {
	        \foreach \y in {1,...,3} {
	          \pgfmathparse{\x+1} \let\z\pgfmathresult
	          \pgfmathparse{\y+1} \let\t\pgfmathresult
	          \draw[fill=gray!40] (\x,\y) rectangle (\x+1,\y+1);
	        }
	      }  
	      \foreach \x in {1,...,4} {
	        \foreach \y in {1,...,3} {
	        \pgfmathparse{\x+1} \let\z\pgfmathresult
	        \pgfmathparse{\y+1} \let\t\pgfmathresult
	        \draw[fill=white] (\x,\y) rectangle (\x+1,\y+1);
	        \draw (\x,\y) -- (\z,\t);
	        \draw (\z,\y) -- (\x,\t);
	      }}\end{tikzpicture})$ and 
	      $t'=(j, \begin{tikzpicture}[scale=0.25] 
	      	      \foreach \x in {1,...,4} {
	        \foreach \y in {1,...,3} {
	          \pgfmathparse{\x+1} \let\z\pgfmathresult
	          \pgfmathparse{\y+1} \let\t\pgfmathresult
	          \draw[fill=gray!40] (\x,\y) rectangle (\x+1,\y+1);
	        }
	      }  
	      \foreach \x/\y in {1/1,2/2,3/1,1/3,3/3,4/2} {
	        \pgfmathparse{\x+1} \let\z\pgfmathresult
	        \pgfmathparse{\y+1} \let\t\pgfmathresult
	        \draw[fill=white] (\x,\y) rectangle (\x+1,\y+1);
	        \draw (\x,\y) -- (\z,\t);
	        \draw (\z,\y) -- (\x,\t);
	      }
	      \end{tikzpicture})$ are not distinguishable.
Indeed, Figure \ref{Fig:sep} denotes all accessible configurations starting from the tableaux of $t$ and $t'$. Every couple of tableaux represent a couple of final states. In this figure, we suppose that $j\cdot w$ is not $m-1$.
If we have $j\cdot w=m-1$, we have two cases to consider:
\begin{enumerate}
\item the cell $(0,0)$ is marked in $t'$. As accessing $m-1$ creates this state, both tableaux are unchanged;
\item the cell $(0,0)$ is not marked in $t'$. In this case, we have to notice that marking this state and saturating the obtained tableau gives the full tableau for $t'$ and so the states are undistinguishable.
\end{enumerate}

\begin{figure}[H]
     \begin{tikzpicture}[scale=0.25] 
\node at (5.2,2.5) [state,ellipse, draw, minimum height=1.5cm, minimum width=4cm] (1){};
	      \foreach \x in {1,...,4} {
	        \foreach \y in {1,...,3} {
	          \pgfmathparse{\x+1} \let\z\pgfmathresult
	          \pgfmathparse{\y+1} \let\t\pgfmathresult
	          \draw[fill=gray!40] (\x,\y) rectangle (\x+1,\y+1);
	        }
	      }  
	      \foreach \x in {1,...,4} {
	        \foreach \y in {1,...,3} {
	        \pgfmathparse{\x+1} \let\z\pgfmathresult
	        \pgfmathparse{\y+1} \let\t\pgfmathresult
	        \draw[fill=white] (\x,\y) rectangle (\x+1,\y+1);
	        \draw (\x,\y) -- (\z,\t);
	        \draw (\z,\y) -- (\x,\t);
	      }}
	      \foreach \x in {6,...,9} {
	        \foreach \y in {1,...,3} {
	          \pgfmathparse{\x+1} \let\z\pgfmathresult
	          \pgfmathparse{\y+1} \let\t\pgfmathresult
	          \draw[fill=gray!40] (\x,\y) rectangle (\x+1,\y+1);
	        }
	      }  
	      \foreach \x/\y in {6/1,7/2,8/1,6/3,8/3,9/2} {
	        \pgfmathparse{\x+1} \let\z\pgfmathresult
	        \pgfmathparse{\y+1} \let\t\pgfmathresult
	        \draw[fill=white] (\x,\y) rectangle (\x+1,\y+1);
	        \draw (\x,\y) -- (\z,\t);
	        \draw (\z,\y) -- (\x,\t);
	      }

\node at (5.2,-4.5) [state, ellipse, draw, minimum height=1.5cm, minimum width=4cm] (2){};
	      \foreach \x in {1,...,4} {
	        \foreach \y in {-6,...,-4} {
	          \pgfmathparse{\x+1} \let\z\pgfmathresult
	          \pgfmathparse{\y+1} \let\t\pgfmathresult
	          \draw[fill=gray!40] (\x,\y) rectangle (\x+1,\y+1);
	        }
	      }  
	      \foreach \x in {1,...,4} {
	        \foreach \y in {-6,...,-4} {
	        \pgfmathparse{\x+1} \let\z\pgfmathresult
	        \pgfmathparse{\y+1} \let\t\pgfmathresult
	        \draw[fill=white] (\x,\y) rectangle (\x+1,\y+1);
	        \draw (\x,\y) -- (\z,\t);
	        \draw (\z,\y) -- (\x,\t);
	      }}
	      \foreach \x in {6,...,9} {
	        \foreach \y in {-6,...,-4} {
	          \pgfmathparse{\x+1} \let\z\pgfmathresult
	          \pgfmathparse{\y+1} \let\t\pgfmathresult
	          \draw[fill=gray!40] (\x,\y) rectangle (\x+1,\y+1);
	        }
	      }  
	      \foreach \x/\y in {6/-6,6/-5,8/-6,7/-4,9/-4,8/-5} {
	        \pgfmathparse{\x+1} \let\z\pgfmathresult
	        \pgfmathparse{\y+1} \let\t\pgfmathresult
	        \draw[fill=white] (\x,\y) rectangle (\x+1,\y+1);
	        \draw (\x,\y) -- (\z,\t);
	        \draw (\z,\y) -- (\x,\t);
	      }
\node at (5.2,-11.5) [ellipse, draw, minimum height=1.5cm, minimum width=4cm] (5){};
	      \foreach \x in {1,...,4} {
	        \foreach \y in {-13,...,-11} {
	          \pgfmathparse{\x+1} \let\z\pgfmathresult
	          \pgfmathparse{\y+1} \let\t\pgfmathresult
	          \draw[fill=gray!40] (\x,\y) rectangle (\x+1,\y+1);
	        }
	      }  
	      \foreach \x in {1,...,4} {
	        \foreach \y in {-13,...,-11} {
	        \pgfmathparse{\x+1} \let\z\pgfmathresult
	        \pgfmathparse{\y+1} \let\t\pgfmathresult
	        \draw[fill=white] (\x,\y) rectangle (\x+1,\y+1);
	        \draw (\x,\y) -- (\z,\t);
	        \draw (\z,\y) -- (\x,\t);
	      }}
	      \foreach \x in {6,...,9} {
	        \foreach \y in {-13,...,-11} {
	          \pgfmathparse{\x+1} \let\z\pgfmathresult
	          \pgfmathparse{\y+1} \let\t\pgfmathresult
	          \draw[fill=gray!40] (\x,\y) rectangle (\x+1,\y+1);
	        }
	      }  
	      \foreach \x/\y in {7/-13,8/-12,9/-13,7/-11,9/-11,6/-12} {
	        \pgfmathparse{\x+1} \let\z\pgfmathresult
	        \pgfmathparse{\y+1} \let\t\pgfmathresult
	        \draw[fill=white] (\x,\y) rectangle (\x+1,\y+1);
	        \draw (\x,\y) -- (\z,\t);
	        \draw (\z,\y) -- (\x,\t);
	      }
  \node at (5.2,-18.5) [ellipse, draw, minimum height=1.5cm, minimum width=4cm] (6){};
	      \foreach \x in {1,...,4} {
	        \foreach \y in {-20,...,-18} {
	          \pgfmathparse{\x+1} \let\z\pgfmathresult
	          \pgfmathparse{\y+1} \let\t\pgfmathresult
	          \draw[fill=gray!40] (\x,\y) rectangle (\x+1,\y+1);
	        }
	      }  
	      \foreach \x in {1,...,4} {
	        \foreach \y in {-20,...,-18} {
	        \pgfmathparse{\x+1} \let\z\pgfmathresult
	        \pgfmathparse{\y+1} \let\t\pgfmathresult
	        \draw[fill=white] (\x,\y) rectangle (\x+1,\y+1);
	        \draw (\x,\y) -- (\z,\t);
	        \draw (\z,\y) -- (\x,\t);
	      }}
	      \foreach \x in {6,...,9} {
	        \foreach \y in {-20,...,-18} {
	          \pgfmathparse{\x+1} \let\z\pgfmathresult
	          \pgfmathparse{\y+1} \let\t\pgfmathresult
	          \draw[fill=gray!40] (\x,\y) rectangle (\x+1,\y+1);
	        }
	      }  
	      \foreach \x/\y in {7/-20,7/-19,9/-20,8/-18,6/-18,9/-19} {
	        \pgfmathparse{\x+1} \let\z\pgfmathresult
	        \pgfmathparse{\y+1} \let\t\pgfmathresult
	        \draw[fill=white] (\x,\y) rectangle (\x+1,\y+1);
	        \draw (\x,\y) -- (\z,\t);
	        \draw (\z,\y) -- (\x,\t);
	      }
\node at (25.2,-1) [ellipse, draw, minimum height=1.5cm, minimum width=4cm] (3){};
	      \foreach \x in {21,...,24} {
	        \foreach \y in {-2.5,...,-0.5} {
	          \pgfmathparse{\x+1} \let\z\pgfmathresult
	          \pgfmathparse{\y+1} \let\t\pgfmathresult
	          \draw[fill=gray!40] (\x,\y) rectangle (\x+1,\y+1);
	        }
	      }  
	      \foreach \x in {21,...,24} {
	        \foreach \y in {-2.5,...,-0.5} {
	        \pgfmathparse{\x+1} \let\z\pgfmathresult
	        \pgfmathparse{\y+1} \let\t\pgfmathresult
	        \draw[fill=white] (\x,\y) rectangle (\x+1,\y+1);
	        \draw (\x,\y) -- (\z,\t);
	        \draw (\z,\y) -- (\x,\t);
	      }}
	      \foreach \x in {26,...,29} {
	        \foreach \y in {-2.5,...,-0.5} {
	          \pgfmathparse{\x+1} \let\z\pgfmathresult
	          \pgfmathparse{\y+1} \let\t\pgfmathresult
	          \draw[fill=gray!40] (\x,\y) rectangle (\x+1,\y+1);
	        }
	      }  
	      \foreach \x/\y in {28/-2.5,29/-1.5,26/-2.5,27/-1.5,27/-0.5,29/-0.5} {
	        \pgfmathparse{\x+1} \let\z\pgfmathresult
	        \pgfmathparse{\y+1} \let\t\pgfmathresult
	        \draw[fill=white] (\x,\y) rectangle (\x+1,\y+1);
	        \draw (\x,\y) -- (\z,\t);
	        \draw (\z,\y) -- (\x,\t);
	      }
  \node at (25.2,-15) [ellipse, draw, minimum height=1.5cm, minimum width=4cm] (7){};
	      \foreach \x in {21,...,24} {
	        \foreach \y in {-16.5,...,-14.5} {
	          \pgfmathparse{\x+1} \let\z\pgfmathresult
	          \pgfmathparse{\y+1} \let\t\pgfmathresult
	          \draw[fill=gray!40] (\x,\y) rectangle (\x+1,\y+1);
	        }
	      }  
	      \foreach \x in {21,...,24} {
	        \foreach \y in {-16.5,...,-14.5} {
	        \pgfmathparse{\x+1} \let\z\pgfmathresult
	        \pgfmathparse{\y+1} \let\t\pgfmathresult
	        \draw[fill=white] (\x,\y) rectangle (\x+1,\y+1);
	        \draw (\x,\y) -- (\z,\t);
	        \draw (\z,\y) -- (\x,\t);
	      }}
	      \foreach \x in {26,...,29} {
	        \foreach \y in {-16.5,...,-14.5} {
	          \pgfmathparse{\x+1} \let\z\pgfmathresult
	          \pgfmathparse{\y+1} \let\t\pgfmathresult
	          \draw[fill=gray!40] (\x,\y) rectangle (\x+1,\y+1);
	        }
	      }  
	      \foreach \x/\y in {29/-16.5,26/-15.5,27/-16.5,28/-15.5,28/-14.5,26/-14.5} {
	        \pgfmathparse{\x+1} \let\z\pgfmathresult
	        \pgfmathparse{\y+1} \let\t\pgfmathresult
	        \draw[fill=white] (\x,\y) rectangle (\x+1,\y+1);
	        \draw (\x,\y) -- (\z,\t);
	        \draw (\z,\y) -- (\x,\t);
	      }
    \node at (45.2,-8) [ellipse, draw, minimum height=1.5cm, minimum width=4cm] (4){};
	      \foreach \x in {41,...,44} {
	        \foreach \y in {-9.5,...,-7.5} {
	          \pgfmathparse{\x+1} \let\z\pgfmathresult
	          \pgfmathparse{\y+1} \let\t\pgfmathresult
	          \draw[fill=gray!40] (\x,\y) rectangle (\x+1,\y+1);
	        }
	      }  
	      \foreach \x in {41,...,44} {
	        \foreach \y in {-9.5,...,-7.5} {
	        \pgfmathparse{\x+1} \let\z\pgfmathresult
	        \pgfmathparse{\y+1} \let\t\pgfmathresult
	      }}
	      	      \foreach \x/\y in {41/-9.5,41/-7.5,43/-9.5,43/-7.5,44/-9.5,44/-7.5} {
	        \pgfmathparse{\x+1} \let\z\pgfmathresult
	        \pgfmathparse{\y+1} \let\t\pgfmathresult
	        \draw[fill=white] (\x,\y) rectangle (\x+1,\y+1);
	        \draw (\x,\y) -- (\z,\t);
	        \draw (\z,\y) -- (\x,\t);
}
	      \foreach \x in {46,...,49} {
	        \foreach \y in {-9.5,...,-7.5} {
	          \pgfmathparse{\x+1} \let\z\pgfmathresult
	          \pgfmathparse{\y+1} \let\t\pgfmathresult
	          \draw[fill=gray!40] (\x,\y) rectangle (\x+1,\y+1);
	        }
	      }  
	      \foreach \x/\y in {46/-9.5,46/-7.5,48/-9.5,48/-7.5,49/-9.5,49/-7.5} {
	        \pgfmathparse{\x+1} \let\z\pgfmathresult
	        \pgfmathparse{\y+1} \let\t\pgfmathresult
	        \draw[fill=white] (\x,\y) rectangle (\x+1,\y+1);
	        \draw (\x,\y) -- (\z,\t);
	        \draw (\z,\y) -- (\x,\t);
	      }
 \path[->]
 (1) edge[bend left] node {a} (2)
     edge node {b} (3)
     edge [bend left] node {c}(4)
 (2) edge[bend left] node {a} (1)
     edge node {b} (5)
     edge  node {c}(4)
 (3) edge[loop,out=65, in=115, looseness=3.5] node[swap] {a} (3)
     edge node {b} (2)
     edge  node {c}(4)
 (5) edge[bend left] node {a} (6)
     edge node {b} (7)
     edge node {c}(4)
 (6) edge[bend left] node {a} (5)
     edge[bend left,out=70, in=110] node {b} (1)
     edge [bend right]  node {c}(4)
 (7) edge[loop,out=-115, in=-65, looseness=3.5] node[swap] {a} (7)
     edge node {b} (6)
     edge  node {c}(4);
    \end{tikzpicture}
\caption{Two undistinguishable tableaux}\label{Fig:sep}
\end{figure}
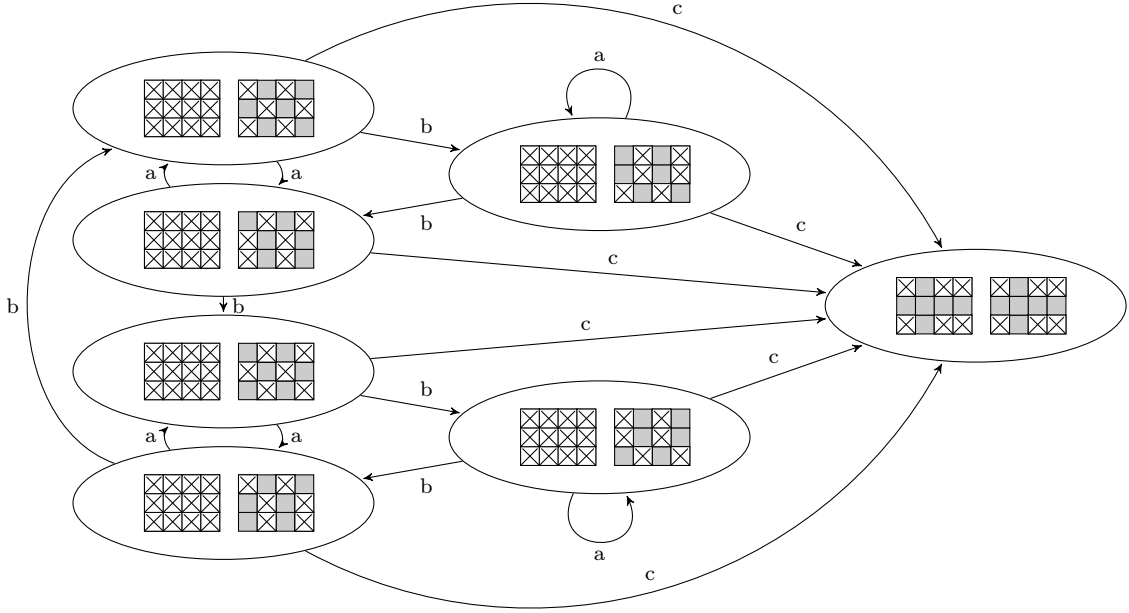

\subsection{$(A\circ B)\cdot C$}
In this section, we propose a $3$-letters Brzozowski witness for $(A\circ B)\cdot C$. 
The tight bounds are  given by
$$\mathrm{sc}_{(*\circ*)\cdot *}(m,n,p)=\mathrm{sc}_{*\cdot *}^\mathrm{f}(\mathrm{sc}_{*\circ *}(m,n),k_\circ,p)$$
where $k_\circ=|(F_A\times Q_B)\circ (Q_A\times F_B)|$ and $\mathrm{sc}^\mathrm{f}_{*\cdot *}$ is defined in \cite{JJS05} by $\mathrm{sc}^\mathrm{f}_{*\cdot *}(m,k,n)=m2^n-k2^{n-1}$ as the state complexity for the catenation of a $m$-state DFA  having $k$ final states with an $n$-state DFA.

Let us  consider the witness 
$$W_{m,n,p}^{(\star\circ\star)\cdot\star}=
(\mathcal X_{m}(a,b,-;\{c\}),\mathcal X_{n}(a,c,-;\{b\}),\mathcal X_{p}(a,-,b;\{c\}))$$ 
for each $m,n,p\geq 3$ and $\circ\in\{\cap,\cup,\oplus\}$ (see Figure \ref{W(AoB).C}).

\begin{figure}[H]
	\centerline{
		\begin{tikzpicture}[node distance=1.2cm, bend angle=25]
			\node[state, initial] (q0) {$0$};
			\node[state] (q1) [right of=q0] {$1$};
			\node[state] (q2) [right of=q1] {$2$};
			\node (etc2) [right of=q2] {$\ldots$};
			\node[state, rounded rectangle] (m-2) [right of=etc2] {${m-2}$};
			\node[state, rounded rectangle, accepting] (m-1) [right of=m-2] {${m-1}$};
			\path[->]
        (q0) edge[bend left] node {$a,b$} (q1)
        (q1) edge[bend left] node {$a$} (q2)
        (q2) edge[bend left] node {$a$} (etc2)
        (etc2) edge[bend left] node {$a$} (m-2)
        (m-2) edge[bend left] node {$a$} (m-1)
        (m-1) edge[out=-115, in=-65, looseness=.2] node[above] {$a$} (q0)
		    (q0) edge[out=115,in=65,loop] node {$c$} (q0)
		    (q1) edge[out=115,in=65,loop] node {$c$} (q1)
		    (q2) edge[out=115,in=65,loop] node {$b,c$} (q2)
		    (m-2) edge[out=115,in=65,loop] node {$b,c$} (m-2)
		    (m-1) edge[out=115,in=65,loop] node {$b,c$} (m-1)
        (q1) edge[bend left] node[above] {$b$} (q0)
			;
			\node[state, initial] (r0) [below of=q0, node distance=4cm] {$0$};
			\node[state] (r1) [right of=r0] {$1$};
			\node[state] (r2) [right of=r1] {$2$};
			\node (etc3) [right of=r2] {$\ldots$};
			\node[state, rounded rectangle] (n-2) [right of=etc3] {${n-2}$};
			\node[state, rounded rectangle, accepting] (n-1) [right of=n-2] {${n-1}$};
			\path[->]
        (r0) edge[bend left] node {$a,c$} (r1)
        (r1) edge[bend left] node {$a$} (r2)
        (r2) edge[bend left] node {$a$} (etc3)
        (etc3) edge[bend left] node {$a$} (n-2)
        (n-2) edge[bend left] node {$a$} (n-1)
        (n-1) edge[out=-115, in=-65, looseness=.2] node[above] {$a$} (r0)
		    (r0) edge[out=115,in=65,loop] node {$b$} (r0)
		    (r1) edge[out=115,in=65,loop] node {$b$} (r1)
		    (r2) edge[out=115,in=65,loop] node {$b,c$} (r2)
		    (n-2) edge[out=115,in=65,loop] node {$b, c$} (n-2)
		    (n-1) edge[out=115,in=65,loop] node {$b, c$} (n-1)
        (r1) edge[bend left] node[above] {$c$} (r0)
			;
			\node[state,initial] (0) [below of=etc2, node distance=2cm] {$0$};
			\node[state] (1) [right of=0] {$1$};
			\node[state] (2) [right of=1] {$2$};
			\node (etc1) [right of=2] {$\ldots$};
			\node[state, rounded rectangle] (p-2) [right of=etc1] {$p-2$};
			\node[state, rounded rectangle, accepting] (p-1) [right of=p-2] {$p-1$};
			\path[->]
        (0) edge[bend left] node {$a$} (1)
        (1) edge[bend left] node {$a$} (2)
        (2) edge[bend left] node {$a$} (etc1)
        (etc1) edge[bend left] node {$a$} (p-2)
        (p-2) edge[bend left] node {$a$} (p-1)
        (p-1) edge[out=-115, in=-65, looseness=.2] node[above] {$a$} (0)
		    (0) edge[out=115,in=65,loop] node {$b, c$} (0)
		    (1) edge[out=115,in=65,loop] node {$c$} (1)
		    (2) edge[out=115,in=65,loop] node {$b,c$} (2)
		    (p-2) edge[out=115,in=65,loop] node {$b, c$} (p-2)
		    (p-1) edge[out=115,in=65,loop] node {$b, c$} (p-1)
        (1) edge[bend left] node[above] {$b$} (0)
			;
    \end{tikzpicture}
  }
  \caption{A 3-letters witness for $(*\circ *) \cdot *$}
  \label{W(AoB).C}
\end{figure}
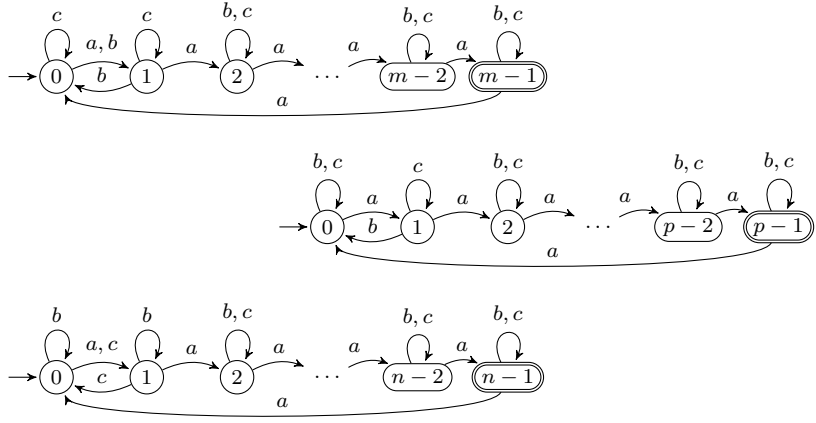

According to the constructions described in Section \ref{Brzozowski witnesses}, we 
define, for each $m,n,p\geq 3$ the automaton
 $$\mathcal 
{R}^{(\star\circ\star)\cdot\star}_{m,n,p}=\left(\mathcal X_{m}(a,b,-;\{c\})\circ\mathcal X_{n}(a,c,-;\{b\})\right)\cdot
\mathcal X_{p}(a,-,b;\{c\}).$$
From Table \ref{combop}, the   states of $\mathcal 
{R}^{(\star\circ\star)\cdot\star}_{m,n,p}$ are indexed by triples $(i,j,S)$ where $0\leq i\leq m-1$,   $0\leq j\leq n-1$ and $S\subset\{0,\dots,p-1\}$. We can notice that not all these states are accessible and it is easy to show the following lemma:

\begin{lemma} 
The states $(i,j,S)$ of $\mathcal {R}^{(\star\circ\star)\cdot\star}_{m,n,p}$ with  $(i,j)$ is a final state of ${\mathcal X}_m\circ {\mathcal X}_n$ and $0\not\in S$ are not accessible.
\end{lemma}
\begin{proof}
Straightforward from Definition \ref{CatenationDFAs}. 
\end{proof}
Let  $\mathrm {Acc}^{(\star\circ\star)\cdot\star}_{m,n,p}$ denote  the states of  $\mathcal {R}^{(\star\circ\star)\cdot\star}_{m,n,p}$ where  $(i,j)$ is a final state of ${\mathcal X}_m\circ {\mathcal X}_n$ implies $0\in S$.
\begin{proposition}
All the states of  $\mathrm {Acc}^{(\star\cap\star)\cdot\star}_{m,n,p}$ are accessible in $\mathcal {R}^{(\star\circ\star)\cdot\star}_{m,n,p}$.
\end{proposition}
\begin{proof}
First, let us notice that the only final state of   ${\mathcal X}_m\cap {\mathcal X}_n$ is $(m-1,n-1)$. 
The proof is by induction on $|S|$. 

If $S=\emptyset$ then $(i,j,\emptyset)$ is accessible from $(0,0,\emptyset)$ by $(ac)^{\mmod{i-j}{m}}a^j$ for $(i,j)\neq (m-1,n-1)$. 

Suppose now that every state $(i,j,S)$ is accessible when $|S|<\alpha$ for $\alpha \geq 1$. We  show that every state $(i,j, S)$ is still accessible when $|S|=\alpha$. Let $S=\{k_1,\ldots ,k_\alpha \}$ with $k_1<k_2<\ldots < k_\alpha$ and consider the four following cases:
\begin{enumerate}[label=(\roman*)]
\item\label{1-A^B} $({m-1},{n-1},S)$ is accessible by $a$ from $({m-2},{n-2},a\cdot (S\setminus \{0\}))$ which is accessible by the induction hypothesis.
\item\label{2-A^B} $(0,0,S)$ with $1\in S$ is accessible by $a$ from $({m-1},{n-1},a\cdot S)$ which is accessible from \ref{1-A^B}.
\item\label{3-A^B} $(0,0,S)$ with $1\not\in S$
\begin{enumerate}[label=(\alph*)]
\item If $0\not\in S$ then $(0,0,S)$ is accessible by $(abc)^{k_1-1}$ from $(0,0,\{1,k_2-k_1+1,\ldots,k_\alpha-k_1+1\})$ which is accessible by  \ref{2-A^B}.
\item If  $0\in S$ then $(0,0,S)$ is accessible by $(abc)^{k_2-1}$ from $(0,0,\{0,1,k_3-k_2+1,\ldots,k_\alpha-k_2+1\})$ which is accessible by  \ref{2-A^B}.
\end{enumerate}
\item $(i,j,S)$ with $(i,j)\not\in \{(0,0),({m-1},{n-1})\}$ is accessible by $(ac)^{\mmod{i-j}{m}}a^j$ from $(0,0,(ac)^{\mmod{i-j}{m}}a^j\cdot S)$ which is accessible by  \ref{2-A^B} or  \ref{3-A^B}.
\end{enumerate}
\end{proof}

\begin{proposition}\label{prop-cup}
All the states of  $\mathrm {Acc}^{(\star\cup\star)\cdot\star}_{m,n,p}$ are accessible.
\end{proposition}
\begin{proof}
First, let us notice that the  final states of   ${\mathcal X}_m\cup {\mathcal X}_n$ are $(\{m-1\}\times Q_n)\cup (Q_m\times \{n-1\})$. 
As $\cup$ is commutative, the state complexity of ${R}^{(\star\cup\star)\cdot\star}_{m,n,p}$ is the same as the state complexity of  ${R}^{(\star\cup\star)\cdot\star}_{n,m,p}$.
So, we  consider, without loss of generality, that $m\geq n$.  The proof is by induction on $|S|$. It is easy to see that   $(i,j,\emptyset)$ is accessible from $(0,0,\emptyset)$ by $(ac)^{i-j}a^j$ if $i\geq j$ or by $(ab)^{j-i}a^i$ if $j>i$.

We have already proved  (Proposition 5 of \cite{CLP16}) that each state $(i,S)$ of $\mathcal 
{R}^{\star\cdot\star}_{m,p}=\mathcal X_{m}(a,b,-;\{c\})\cdot \mathcal X_{p}(a,-,b;\{c\})$ with $i=m-1\Rightarrow 0\in S$ is accessible by a word $w$. Let us now show that each state $(i,j,S)$ of $\mathrm {Acc}^{(\star\cup\star)\cdot\star}_{m,n,p}$ is also accessible. As $S$ is not empty, the word $w$ is composed of at least $m-1$ letters $a$. As $m\geq n\geq j$, we can deduce a word $w'$ from $w$ by replacing all but the last $j$ $a$ by $ac$. Let $w'=u\cdot v$ where $u$ is the prefix of $w'$ where each $a$ has been replaced by $ac$. Then, in $ \mathcal X_{n}(a,c,-;\{b\})$ we have $0\cdot u=0$ and $0\cdot v=j$ so $(0,0,S)\cdot w'=(i,j,S)$.
\end{proof}

\begin{proposition}
All the states of  $\mathrm {Acc}^{(\star\oplus\star)\cdot\star}_{m,n,p}$ are accessible.
\end{proposition}
\begin{proof}
This case is nearly the same as the one of Proposition \ref{prop-cup}. The only difference to consider is the state $({m-1},{n-1}, S)$ with $0\not\in S$ which has to be accessible for any subset $S$ since $(m-1,n-1)$ is the only state which is final for  ${\mathcal X}_m\cup {\mathcal X}_n$ and not for  ${\mathcal X}_m\oplus {\mathcal X}_n$. Let us consider without loss of generality that $m\geq n$. 
First, the state  $({m-1},{n-1}, \emptyset)$ is accessible by $a$ from $({m-2},{n-2}, \emptyset)$. By Proposition \ref{prop-cup}, any state $(i,j,S)$, with  $i=m-1$ or $j=n-1$ implies $0\in S$,  is accessible. So  any $({m-1},{n-1},S)$ with $0\not\in S$ is accessible by $a$ from  $({m-2},{n-2},a\cdot S)$,  which is accessible.
\end{proof}

\begin{lemma}\label{noneqcirc}
	The states of $\mathcal X_{m}(a,b,-;\{c\})\circ \mathcal X_{n}(a,c,-;\{b\})$ are pairwise nonequivalent
\end{lemma}
\begin{proof}
Let $(i,j)\neq (i',j')$.
	Suppose first $\circ=\cup$ or $\circ=\oplus$. 
	Without lost of generalities we assume $i\neq i'$ (the other case $j\neq j'$ being symmetric, it is obtained by replacing $c$ by $b$). As any $w\in \{a,b,c\}^*$ induces a permutation and according to Lemma \ref{lm-inverse}, we have
	\begin{equation}
	(i,j)\cdot a^{-j}(ac)^{j-i-1}=(m-1,0)
\end{equation} which is a final state while  
\begin{equation}\begin{array}{rcl}(i',j')\cdot  a^{-j}(ac)^{j-i-1}=(i'-i-1,0)\end{array}\end{equation} is not a final state because $i'\neq i$.\\
Now suppose $\circ=\cap$. Without lost of generalities we assume $i\neq i'$ (the other case $j\neq j'$ being symmetric, it is obtained by replacing $c$ by $b$). We have
\begin{equation}
	(i,j)\cdot a^{-j}(ac)^{n-i+j}a^{n-1}=
	(i-j,0)\cdot (ac)^{-n-i+j}a^{n-1}=(-n,0)\cdot a^{n-1}=(m-1,n-1)
\end{equation}
while
\begin{equation}\begin{array}{rcl}
	(i',j)\cdot a^{-j}(ac)^{n-i+j}a^{n-1}=
	(i'-j,0)\cdot (ac)^{-n-i+j}a^{n-1}&=&(i'-i-n,0)\cdot a^{n-1}\\&=&(i'-i-1,n-1)\neq(m-1,n-1).\end{array}
\end{equation}
\end{proof}
\begin{proposition}
	The states of $\mathcal R^{(\star\circ\star)\cdot\star}_{m,n,p}$ are pairwise nonequivalent.
\end{proposition}
\begin{proof}
	Let $(i,j,S)\neq (i,j,S')$. Suppose that $S\neq S'$. Without loss of generality, one assume $S\neq \emptyset$. Let $k=\max S\setminus S'$.
	Then $(i,j,S)\cdot a^{p-k-1}$ is final while $(i',j',S')\cdot a^{p-k-1}$ is not final.
	
	Now suppose $(i,j)\neq (i',j')$. From Lemma \ref{noneqcirc} there exists $w$ such that $(i,j)\cdot w$ is final in 
	$\mathcal X_{m}(a,b,-;\{c\})$ $\circ \mathcal X_{n}(a,c,-;\{b\})$ and $(i',j')\cdot w$ is not final. 
	Remarking that $ s\cdot a^{\mmod{1-p}{mn}}b^{2}a^{p-1}=s$ for any state $s$ of $\mathcal X_{m}(a,b,-;\{c\})\circ \mathcal X_{n}(a,c,-;\{b\})$,
	we find that $(i,j,S)\cdot wa^{\mmod{1-p}{mn}}b^{2}a^{p-1}=(i_{f},j_{f},S_{3})$ where $0\in S_{1}$ and 
	$(i',j',S')\cdot wa^{\mmod{1-p}{mn}}b^{2}a^{p-1}=(i_{n},j_{n},S'_{3})$ where $0\not\in S'_{1}$. 
	Indeed,
	$(i,j,S)\cdot w=(i_{f},j_{f},S_{1})$ with $(i_{f},j_{f})$ final in $\mathcal X_{m}(a,b,-;\{c\})\circ \mathcal X_{n}(a,c,-;\{b\})$.
	Hence, $(i_{f},j_{f},S_{1})\cdot a^{[1-p]_{mn}}b^{2}=(i_{f}+1-p,j_{f}+1-p,S_{2})$ where $1\not\in S_{2}$. So,
	$(i_{f}+1-p,j_{f}+1-p,S_{2}).a^{p-1}=(i_{f},j_{f},S_{3})$ where $0\in S_{3}$ because $(i_{f},j_{f})$ is final in 
	$\mathcal X_{m}(a,b,-;\{c\})\circ \mathcal X_{n}(a,c,-;\{b\})$. \\
	On the other hand, $(i',j',S)\cdot w=(i_{n},j_{n},S'_{1})$  with $(i_{n},j_{n})$ non final in 
	$\mathcal X_{m}(a,b,-;\{c\})\circ \mathcal X_{n}(a,c,-;\{b\})$. 
	Hence, $(i_{n},j_{n},S'_{1})\cdot a^{\mmod{1-p}{mn}}b^{2}=(i_{n}+1-p,j_{n}+1-p,S'_{2})$ where $1\not\in S'_{2}$. So,
	$(i_{n}+1-p,j_{n}+1-p,S'_{2})\cdot a^{p-1}=(i_{n},j_{n},S'_{3})$ where $0\in S'_{3}$ because $(i_{n},j_{n})$ is not final in 
	$\mathcal X_{m}(a,b,-;\{c\})\circ \mathcal X_{n}(a,c,-;\{b\})$.\\
	To summarize we have $S_{3}\neq S'_{3}$. 
	So the problem reduces to the first case ($S\neq S')$. This shows the result.
\end{proof}

\begin{theorem}
	The state complexity  $\mathrm {sc}_{(*\circ *)\cdot *)}(m,n,p)$ is $(mn-k)2^{p}k2^{p-1}$ where $k=1$ for $\cap$, $k=m+n-1$ for $\cup$ and $k=m+n-2$ for $\oplus$.  In  all cases, the bound is reached by the three-letters witness $W_{m,n,p}^{(\star\circ\star)\cdot\star)}$.
\end{theorem}

\subsection{$(A\cdot B)\circ C$}

In this section, we propose a $2$-letters Brzozowski witness for $(A\cdot B)\circ C$. 
The tight bounds are  given by
$$\mathrm{sc}_{(*\cdot*)\circ *}(m,n,p)=\mathrm{sc}_{*\circ *}(\mathrm{sc}_{*\cdot *}(m,n),p).$$

Let us  consider the witness 
$$W_{m,n,p}^{(\star\cdot\star)\circ\star}=
(\mathcal X_{m}(a,b,-;\emptyset),\mathcal X_{n}(a,-,b;\emptyset),\mathcal X_{p}(b,a,-;\emptyset))$$ 
for each $m,n,p\geq 3$ and $\circ\in\{\cap,\cup,\oplus\}$ (see Figure \ref{W(A.B)oC}).

\begin{figure}[H]
	\centerline{
		\begin{tikzpicture}[node distance=1.2cm, bend angle=25]
			\node[state, initial] (q0) {$0$};
			\node[state] (q1) [right of=q0] {$1$};
			\node[state] (q2) [right of=q1] {$2$};
			\node (etc2) [right of=q2] {$\ldots$};
			\node[state, rounded rectangle] (m-2) [right of=etc2] {${m-2}$};
			\node[state, rounded rectangle, accepting] (m-1) [right of=m-2] {${m-1}$};
			\path[->]
        (q0) edge[bend left] node {$a,b$} (q1)
        (q1) edge[bend left] node {$a$} (q2)
        (q2) edge[bend left] node {$a$} (etc2)
        (etc2) edge[bend left] node {$a$} (m-2)
        (m-2) edge[bend left] node {$a$} (m-1)
        (m-1) edge[out=-115, in=-65, looseness=.2] node[above] {$a$} (q0)
		    (q2) edge[out=115,in=65,loop] node {$b$} (q2)
		    (m-2) edge[out=115,in=65,loop] node {$b$} (m-2)
		    (m-1) edge[out=115,in=65,loop] node {$b$} (m-1)
        (q1) edge[bend left] node[above] {$b$} (q0);
			\node[state, initial, right of =m-1, node distance=2cm] (r0) {$0$};
			\node[state] (r1) [right of=r0] {$1$};
			\node[state] (r2) [right of=r1] {$2$};
			\node (etc3) [right of=r2] {$\ldots$};
			\node[state, rounded rectangle] (n-2) [right of=etc3] {${n-2}$};
			\node[state, rounded rectangle, accepting] (n-1) [right of=n-2] {${n-1}$};
			\path[->]
        (r0) edge[bend left] node {$a$} (r1)
        (r1) edge[bend left] node {$a$} (r2)
        (r2) edge[bend left] node {$a$} (etc3)
        (etc3) edge[bend left] node {$a$} (n-2)
        (n-2) edge[bend left] node {$a$} (n-1)
        (n-1) edge[out=-115, in=-65, looseness=.2] node[above] {$a$} (r0)
		    (r0) edge[out=115,in=65,loop] node {$b$} (r0)
		    (r2) edge[out=115,in=65,loop] node {$b$} (r2)
		    (n-2) edge[out=115,in=65,loop] node {$b$} (n-2)
		    (n-1) edge[out=115,in=65,loop] node {$b$} (n-1)
        (r1) edge[bend left] node[above] {$b$} (r0);
			\node[state, initial, below of =etc2, node distance=2cm] (0) {$0$};
			\node[state] (1) [right of=0] {$1$};
			\node[state] (2) [right of=1] {$2$};
			\node (etc1) [right of=2] {$\ldots$};
			\node[state, rounded rectangle] (p-2) [right of=etc1] {${p-2}$};
			\node[state, rounded rectangle, accepting] (p-1) [right of=p-2] {${p-1}$};
			\path[->]
        (0) edge[bend left] node {$a,b$} (1)
        (1) edge[bend left] node {$b$} (2)
        (2) edge[bend left] node {$b$} (etc1)
        (etc1) edge[bend left] node {$b$} (p-2)
        (p-2) edge[bend left] node {$b$} (p-1)
        (p-1) edge[out=-115, in=-65, looseness=.2] node[above] {$b$} (0)
		    (2) edge[out=115,in=65,loop] node {$a$} (2)
		    (p-2) edge[out=115,in=65,loop] node {$a$} (p-2)
		    (p-1) edge[out=115,in=65,loop] node {$a$} (p-1)
        (1) edge[bend left] node[above] {$a$} (0);
    \end{tikzpicture}
  }
  \caption{A 2-letters witness for $(*\cdot *) \circ *$}
  \label{W(A.B)oC}
\end{figure}
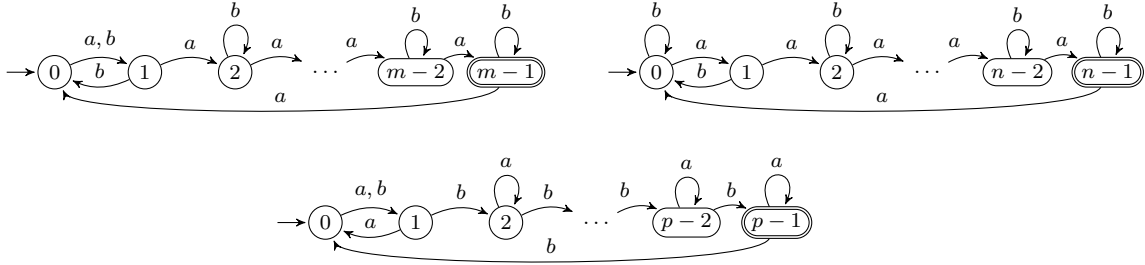

According to the constructions described in Section \ref{Brzozowski witnesses}, we 
define, for each $m,n,p\geq 3$ the automaton
 $$\mathcal 
{R}^{(\star\cdot\star)\circ\star}_{m,n,p}=\left(\mathcal X_{m}(a,b,-;\emptyset)\cdot\mathcal X_{n}(a,-,b;\emptyset)\right)\circ
\mathcal X_{p}(b,a,-;\emptyset).$$
From Table \ref{combop}, the   states of $\mathcal 
{R}^{(\star\cdot\star)\circ\star}_{m,n,p}$ are indexed by triples $(i,S,k)$ where $0\leq i\leq m-1$,   $0\leq k\leq p-1$ and $S\subset\{0,\dots,n-1\}$.

Let  $\mathrm {Acc}^{(\star\cdot\star)\circ\star}_{m,n,p}$ denote  the states of  $\mathcal {R}^{(\star\cdot\star)\circ\star}_{m,n,p}$ where  $i=m-1$ implies $0\in S$.
\begin{proposition}
All the states of  $\mathrm {Acc}^{(\star\cdot\star)\circ\star}_{m,n,p}$ are accessible in $\mathcal {R}^{(\star\cdot\star)\circ\star}_{m,n,p}$.
\end{proposition}

\begin{proof}


Let us first show that each state $(0,\emptyset, k)$ is accessible from $(0,\emptyset, 0)$.
\begin{itemize}
\item if $k$ is even then $(0,\emptyset,0)\xrightarrow{b^k}(0,\emptyset, k)$.
\item if $k>1$ and $k$ is odd then $(0,\emptyset,0)\xrightarrow{ab^{k-1}}(0,\emptyset, k)$
\item if $k=1$ and $k$ is odd then  $(0,\emptyset,0)\xrightarrow{b^{p+1}}(0,\emptyset, 1)$
\item if $k=1$ and $k$ is even then  $(0,\emptyset,0)\xrightarrow{b^2}(0,\emptyset, 2)\xrightarrow{a}(1,\emptyset, 2)\xrightarrow{b^{p-1}}(0,\emptyset, 1)$
\end{itemize}
We know  that every state $(i,S)$ of $\mathcal X_{m}(a,b,-;\emptyset)\cdot\mathcal X_{n}(a,-,b;\emptyset)$ is
 accessible (Proposition $4$ of  \cite{CLP16}) which means that  there exists a word $w$ such that $(0,\emptyset)\cdot w=(i,S)$. Let $(i,S,k)$ be a state of  $\mathrm {Acc}^{(\star\cdot\star)\circ\star}_{m,n,p}$ and let $w$ be such that  $(0,\emptyset)\cdot w=(i,S)$. The word $w$ acts as a permutation on the states of $\mathcal X_{p}(b,a,-;\emptyset)$. So, there exists $k'\in \mathcal X_{p}(b,a,-;\emptyset)$ such that $k'\cdot w=k$. Thus $(i,S,k)$ is accessible from $(0,\emptyset,k')$ by $w$. As $(0,\emptyset,k')$ is accessible, we  conclude that $(i,S,k)$ is accessible.

\end{proof}

%
To prove the pairwise non equivalence, we need the following lemmas. 
\begin{lemma}\label{l.final}
	Let $m,n>2$. Any state $(i,S)$ of $\mathcal X_{m}(a,b,-;\emptyset)\cdot\mathcal X_{n}(a,-,b;\emptyset)$ is final if and only if
	$(i,S)\cdot b$ is final.
\end{lemma}
\begin{proof}
Let us suppose that $(i,S)$ is final. Then $n-1\in S$. As $(n-1)\cdot b =n-1$ in $\mathcal X_{n}(a,-,b;\emptyset)$, we have $n-1\in S\cdot b$ which means that  $(i,S)\cdot b$ is a final state of $\mathcal X_{m}(a,b,-;\emptyset)\cdot\mathcal X_{n}(a,-,b;\emptyset)$. The converse is straightforward.
\end{proof}
\begin{lemma}\label{l.nonequivalent}
	All the  states  of  $\mathcal X_{m}(a,b,-;\emptyset)\cdot\mathcal X_{n}(a,-,b;\emptyset)$ are co-accessible and for any state $(i,S)$ there exists a word $w$ such that $(i,S)\cdot w$ is not a final state.
\end{lemma}
\begin{proof}
	The minimality of $\mathcal X_{m}(a,b,-;\emptyset)\cdot\mathcal X_{n}(a,-,b;\emptyset)$ is proved in  \cite{CLP16}.
	Since the automaton is minimal, it has at most one non co-accessible state. So this state is invariant by the action of $a$ and $b$.
	Remarking that for each state $(i,S)$, one has $(i,S)\cdot a=(i+1,S')$, we prove that there is no non  co-accessible state. In the same way, 
	the pairwise non equivalence implies that there exists at most one state $(i,S)$ such that $(i,S)\cdot w$ is final for any $w$.
	 The same argument as for the co-accessibility allows us to conclude.
\end{proof}
\begin{proposition}
	The states belonging to $Acc^{(\star\cdot\star)\circ\star}_{m,n,p}$  are pairwise non equivalent.
\end{proposition}

\begin{proof}
	
Let $(i,S,k)\neq (i',S',k')$ be a state of ${\cal R}^{(\star\cdot\star)\circ\star}_{m,n,p}$. Suppose first that $(i,S)=(i',S')$ then $k\neq k'$.
From Lemma \ref{l.nonequivalent},
there exists a word $u$ such that $(i,S)\cdot u$ is final and another word $v$ such that $(i,S)\cdot v$ is not final in 
$\mathcal X_{m}(a,b,-;\emptyset)\cdot\mathcal X_{n}(a,-,b;\emptyset)$. 
If  $\circ=\cap${}, let $k_1=k\cdot u$ and $k'_1=k'\cdot u$ in   $\mathcal X_{p}(b,a,-;\emptyset)$ and $w= ub^{p-1-k_1}$. Then $(i,S)\cdot w$ is final in  $\mathcal X_{m}(a,b,-;\emptyset)\cdot\mathcal X_{n}(a,-,b;\emptyset)$ (Lemma \ref{l.final}) and $w$ separates the two states because  
 $(i,S,k)\cdot w= ((i,S)\cdot w,p-1)$
  is final in
 ${\cal R}^{(\star\cdot\star)\circ\star}_{m,n,p}$ 
 while $(i,S,k')\cdot w=((i,S)\cdot w,k'\cdot  ub^{p-1-k_1})= ((i,S)\cdot w,k'_1-k_1+p-1)$ is not final. 
 Suppose now $\circ=\oplus$ or $\circ=\cup$.
 The word $w=vb^{p-1-k'_1}$ separates the two states because Lemma \ref{l.final} implies that
 $(i,S,k)\cdot w=((i,S)\cdot w,k-k'_1+p-1)$ is not final  while $(i,S,k')\cdot w= ((i,S)\cdot w,p-1)$ is final.\\
 Suppose now $(i,S)\neq (i',S')$. From Lemma \ref{l.nonequivalent}, there exists a word $u$ separating the two states in $\mathcal X_{m}(a,b,-;\emptyset)\cdot\mathcal X_{n}(a,-,b;\emptyset)$.
  Without loss of generalities
 we assume $(i,S)\cdot u$ is final and $(i',S')\cdot u$ is non final. Let $k_1=k\cdot u$ and $k'_1=k'\cdot u$ in $\mathcal X_{p}(b,a,-;\emptyset)$ and let $w=ub^{p-1-k_1}$. 
 If  $\circ=\cap$ then, by Lemma \ref{l.final}, the state $(i,S,k)\cdot w$ is final while $(i',S',k')\cdot w$ is not final.
 If $\circ=\oplus$ or $\circ=\cup$ then, we consider $\delta$ such that $\mmod{k_1+\delta}{p},\mmod{k'_1+\delta}{p}\neq p-1$. Such an integer exists because $p>2$. 
 By Lemma \ref{l.final}, the state $(i,S,k)\cdot wb^{\delta}$ is final while $(i',S',k')\cdot wb^{\delta}$ is not final.
 \\ In all the cases, the states $(i,S,k)$ and $(i',S',k')$ are non equivalent. 
 \end{proof}

 \begin{theorem}
	The state complexity  $\mathrm {sc}_{(*\cdot *)\circ *)}(m,n,p)$ is $((m-1)2^{n}+2^{n-1})p$.  The bound is reached by the two-letters witness $W_{m,n,p}^{(\star\circ\star)\cdot\star)}$.
\end{theorem}

\section{Conclusion}

For any $3$-ary operator involving a catenation and/or a boolean operator, we give a Brzozowski witness. In many cases, it allows us to improve some conjectures by Brzozowski.
Nevertheless, many questions remain to be investigated. In particular, the optimality of the size of the alphabet remains a difficult problem which may require the development of  algebraic and combinatorial tools. Indeed, in all the constructions the states are labelled by combinatorial objects as tableaux, sets, \textit{etc}. The letters can be seen as operators acting on these objects and generate a semigroup.
All the elements and the actions can be combinatorially described.  This problem has to be restated to settle in well the theory of finite semigroups. 

\bibliography{../COMMONTOOLS/biblio,../COMMONTOOLS/bibjg}

\end{document}